\documentclass[11pt]{article}
\usepackage{graphicx,hyperref,color}
\usepackage{amsmath,amssymb,amsthm}
\usepackage{enumerate}
\usepackage{fullpage}

\newtheorem{theorem}{Theorem}

\newtheorem{lemma}[theorem]{Lemma}
\newtheorem{claim}[theorem]{Claim}

\newtheorem{observation}[theorem]{Observation}

\newcommand{\dm}{\mathrm{diam}}
\newcommand{\edm}{\mathrm{ediam}}

\newcommand{\mf}{\mathcal{F}}
\newcommand{\avgH}{\sigma_H}
\newcommand{\savgH}{|V(H)|\sqrt{\log |V(H)|}}
\newcommand*\trunc[2]{\lfloor {#1} \rfloor_{#2}}
\newcommand*\rfrac[2]{{}^{#1}\!/_{#2}}

\DeclareMathOperator{\mst}{\textsc{MST}}
\DeclareMathOperator{\MST}{\mathrm{MST}}
\DeclareMathOperator{\poly}{\mathrm{poly}}

\date{}

\title{Minor-free graphs have light spanners}
\author{Glencora Borradaile \\Oregon State University \\ \texttt{glencora@engr.oregonstate.edu}  \and Hung Le \\Oregon State University \\ \texttt{lehu@onid.oregonstate.edu}  \and Christian Wulff-Nilsen \\University of Copenhagen \\ \texttt{koolooz@di.ku.dk }}
\begin{document}
\maketitle
\begin{abstract}
  We show that every $H$-minor-free graph has a light $(1+\epsilon)$-spanner, resolving an open problem of Grigni and Sissokho~\cite{GS02} and proving a conjecture of Grigni and Hung~\cite{GH12}. Our lightness bound is \[O\left(\frac{\avgH}{\epsilon^3}\log  \frac{1}{\epsilon}\right)\] where $\avgH = \savgH$ is the sparsity coefficient of $H$-minor-free graphs.  That is, it has a practical dependency on the size of the minor $H$. Our result also implies that the polynomial time approximation scheme (PTAS) for the Travelling Salesperson Problem (TSP) in $H$-minor-free graphs by Demaine, Hajiaghayi and Kawarabayashi~\cite{DHK11} is an \emph{efficient} PTAS whose running time is $2^{O_H\left(\frac{1}{\epsilon^4}\log  \frac{1}{\epsilon}\right)}n^{O(1)}$ where $O_H$  ignores dependencies on the size of $H$. Our techniques significantly deviate from existing lines of research on spanners for $H$-minor-free graphs, but build upon the work of Chechik and Wulff-Nilsen for spanners of general graphs~\cite{CW16}. 

\end{abstract}

\section{Introduction}

Peleg and Sch\"{a}ffer~\cite{PS89} introduced $t$-spanners of graphs as a way to sparsify graphs while approximately preserving pairwise distances between vertices. A \emph{$t$-spanner} of a graph $G$ is a subgraph $S$ of $G$ such that $d_S(x,y) \leq t\cdot d_G(x,y)$ for all vertices $x,y$ \footnote{We use standard graph terminology, which can be found in Appendix~\ref{sec:prel}.}. Two parameters of $t$-spanners that are of interest are their \emph{sparsity} and \emph{lightness}. The \emph{sparsity} of $S$ is the ratio of the number of edges to the number of vertices of $S$. The \emph{lightness} of $S$ is the ratio of the total weight of the edges of $S$ to the weight of an $\mst$ of $G$; generally, we assume that $\mst(G) \subseteq S$ (and so $\mst(S) = \MST(G)$). Here, we are concerned with the lightness of $(1+\epsilon)$-spanners, where $\epsilon<1$, and so we refer to $(1+\epsilon)$-spanners simply as {\em spanners}. 

We say that a spanner is \emph{light} if the lightness does not depend on the number of vertices in the graph.  Grigni and  Sissokho~\cite{GS02} showed that $H$-minor-free graphs have spanners of lightness 
\begin{equation}
O\left({\textstyle\frac{1}{\epsilon}}\avgH\log n\right).\label{eq:notlight}
\end{equation}
where $\avgH = \savgH$ is the sparsity coefficient of $H$-minor-free graphs; namely that an $H$-minor-free  graph of $n$ vertices has $O(\savgH n)$ edges\footnote{This bound is tight~\cite{thomason01}.}~\cite{Mader68}. 
Later Grigni and Hung~\cite{GH12}, in showing that graphs of bounded pathwidth have light spanners, conjectured that  $H$-minor-free graphs also have light spanners; that is, that the dependence on $n$ can be removed from the lightness above.  In this paper, we resolve this conjecture positively, proving:
\begin{theorem}\label{thm:main}
  Every $H$-minor-free graph $G$ has a $(1+\epsilon)$-spanner of lightness
  \begin{equation}
O\left(\frac{\avgH}{\epsilon^3}\log  \frac{1}{\epsilon}\right).\label{eq:light}
\end{equation}
\end{theorem}
Our algorithm consists of a reduction phase and a greedy phase.  In the reduction phase, we adopt a technique of Chechik and Wulff-Nilsen~\cite{CW16}: edges of the graph are subdivided and their weights are rounded and scaled to guarantee that every $\mst$-edge has unit weight and we include all very low weight edges in the spanner (Appendix~\ref{app:reduction}).  In the greedy phase, we use the standard greedy algorithm for constructing a spanner to select edges from edges of the graph not included in the reduction phase (Appendix~\ref{sec:greed}).

As a result of the reduction phase, our spanner is not the ubiquitous {\em greedy spanner}.  However, since Filtser and Solomon have shown that greedy spanners are (nearly) optimal in their lightness~\cite{FS16}, our result implies that the greedy spanner for $H$-minor-free graphs is also light.

\subsection{Implication: Approximating TSP}

Light spanners have been used to give PTASes, and in some cases efficient PTASes, for the traveling salesperson problem (TSP) on various classes of graphs.  A PTAS, or polynomial-time approximation scheme, is an algorithm which, for a fixed error parameter $\epsilon$, finds a solution whose value is within $1\pm\epsilon$ of optimal in polynomial time. A PTAS is \emph{efficient} if its running time is $f(\epsilon)n^{O(1)}$ where $f(\epsilon)$ is a function of $\epsilon$. Rao and Smith~\cite{RS98} used light spanners of Euclidean graphs to give an EPTAS for Euclidean TSP. Arora, Grigni, Karger, Klein and Woloszyn~\cite{AGKKW98} used light spanners of planar graphs, given by Alth\"{o}fer, Das, Dobkin, Joseph and Soares~\cite{ADDJS93}, to design a PTAS for TSP in planar graphs with running time $n^{O(\frac{1}{\epsilon^2})}$. Klein~\cite{Klein05} improved upon this running time to $2^{O(\frac{1}{\epsilon^2})}n$ by modifying the PTAS framework, using the same light spanner. Borradaile, Demaine and Tazari generalized Klein's EPTAS to bounded genus graphs~\cite{BDT14}. 

In fact, it was in pursuit of a PTAS for TSP in $H$-minor-free graphs that Grigni and  Sissokho discovered the logarithmic bound on lightness (Equation~\eqref{eq:notlight}); however, the logarithmic bound implies only a quasi-polynomial time approximation scheme (QPTAS) for TSP~\cite{GS02}.
Demaine, Hajiaghayi and Kawarabayashi~\cite{DHK11} used Grigni and  Sissokho's spanner to give a PTAS for TSP in $H$-minor-free graphs with running time $n^{O(\poly(\frac{1}{\epsilon}))}$; that is, {\em not} an efficient PTAS.  However, Demaine, Hajiaghayi and Kawarabayashi's PTAS is {\em efficient if the spanner used is light}. Thus, the main result of this paper implies an efficient PTAS for TSP in  $H$-minor-free graphs.

\subsection{Techniques}

In proving the lightness of spanners in planar graphs~\cite{ADDJS93} and bounded genus graphs~\cite{Grigni00}, the  embedding of the graph was heavily used. Thus, it is natural to expect that showing minor-free graphs have light spanners would rely on the 
decomposition theorem of minor-free graphs by Robertson and Seymour~\cite{RS03}, which shows that graphs excluding a fixed minor can be decomposed into the clique-sum of graphs {\em nearly} embedded on surfaces of fixed genus.  Borradaile and Le~\cite{BL17} use this decomposition theorem to show that if graphs of bounded treewidth have light spanners, then $H$-minor-free graphs also have light spanners.  As graphs of bounded treewidth are generally regarded as \emph{easy instances} of $H$-minor-free graphs, it may be possible to give a simpler proof of lightness of spanners for $H$-minor-free graphs using this implication.

However, relying on the Robertson and Seymour decomposition theorem generally results in constants which are galactic in the size of the the minor~\cite{Lipton10,johnson87}.  In this work, we take a different approach which avoids this problem.  Our method is inspired from the recent work of  Chechik and Wulff-Nilsen~\cite{CW16} on spanners for general graphs which uses an iterative super-clustering technique~\cite{ALGP89,EP04}. Using the same technique in combination with amortized analysis, we show that $H$-minor-free graphs not only have light spanners, but also that the dependency of the lightness on $\epsilon$ and $|V(H)|$ is practical (Equation~\eqref{eq:light}).

At a high level, our proof shares several ideas with the work of Chechik and Wulff-Nilsen~\cite{CW16} who prove that (general) graphs have $(2k-1)\cdot(1+\epsilon)$-spanners with lightness $O_\epsilon(n^{1/k})$, removing a factor of $k/\log k$ from the previous best-known bound and matching Erd\H{o}s's girth conjecture~\cite{Erdos64} up to a $1+\epsilon$ factor.  Our work differs from Chechik and Wulff-Nilsen  in two major aspects. First, Chechik and Wulff-Nilsen reduce their problem down to a single hard case where the edges of the graph have weight at most  $g^k$ for some constant $g$.  In our problem, we must partition the edges according to their weight along a logarithmic scale and deal with each class of edges separately.
Second, we must employ the fact that $H$-minor-free graphs (and their minors) are sparse in order to get a lightness bound that does not depend on $n$.  

\subsection{Future directions}

Since we avoid relying on Robertson and Seymour's decomposition theorem and derive bounds using only the sparsity of graphs excluding a fixed minor, it is possible this technique could be extended to related {\em spanner-like} constructions that are used in the design of PTASes for connectivity problems.  Except for TSP, many connectivity problems~\cite{BDT14} have PTASes for bounded genus graphs but are not known to have PTASes for $H$-minor-free graphs -- for example, subset TSP and Steiner tree.   The PTASes for these problems rely on having a light subgraph that approximates the optimal solution within $1+\epsilon$  (and hence is {\em spanner-like}).  The construction of these subgraphs, though, rely heavily on the embedding of the graph on a surface and since the Robertson and Seymour decomposition gives only a weak notion of embedding for $H$-minor-free graphs, pushing these PTASes beyond surface embedded-graphs does not seem likely.  The work of this paper may be regarded as a first step toward designing {\em spanner-like} graphs for problems such as subset TSP and Steiner tree that do not rely on the embedding.

\section{Bounding the lightness of a $(1+\epsilon)$-spanner} \label{sec:bound-lightness}
As we already indicated, we start with a reduction that allows us to assume that the edges of the $\mst$ of the graph each have unit weight. (For details, see Appendix~\ref{app:reduction}.)
For simplicity of presentation, we will also assume that the spanner is a greedy $(1+s\cdot \epsilon))$-spanner for a sufficiently large constant $s$; this does not change the asymptotics of our lightness bound. 

Herein, we let $S$ be the edges of a greedy $(1+s \cdot \epsilon)$-spanner of graph $G$ with an $\mst$ having edges all of unit weight. We simply refer to $S$ as the \emph{spanner}. The greedy spanner considers the edges in non-decreasing order of weights and adds an edge $xy$ if $(1+s \cdot \epsilon)w(xy)$ is at most the $x$-to-$y$ distance in the current spanner (see Appendix~\ref{sec:greed} for a review).

We partition the edges of $S$ according to their weight as it will be simpler to bound the weight of subsets of $S$.  Let $J_0$ be the edges of $S$ of weight in the range $[1,{1\over\epsilon})$; note that $\mst \subseteq J_0$ and, since $G$ has $O(\avgH n)$ edges  and $w(\mst) = n-1$, 
\begin{equation}
w(J_0) = O(\avgH n/\epsilon) = O\left( {\avgH\over \epsilon} w(\mst) \right)\label{eq:J0}
\end{equation}
Let $\Pi_j^i$ be the edges of $S$ of weight in the range $[{2^j\over\epsilon^i},{2^{j+1}\over\epsilon^i})$ for every $i \in \mathbb{Z}^+$ and $j \in \{0,1, \ldots, \lceil \log \frac{1}{\epsilon} \rceil\}$.  Let $J_j = \cup_i \Pi_j^i$.  We will prove that 
\begin{lemma}\label{lem:main-lightness}
There exists a set of spanner edges $B$ such that $w(B) = O(\frac{1}{\epsilon^2}w(\mst))$ 
and for every $j \in \{0,\ldots, \lceil \log \frac{1}{\epsilon} \rceil\}$, 
\[w\left(\mst \cup (J_j\setminus B) \right)= O\left(\frac{\avgH}{\epsilon^3}\right) w(\mst).\]
\end{lemma}
Combined with Equation~(\ref{eq:J0}), Lemma~\ref{lem:main-lightness} gives us
\begin{equation*}
w(S) = w(B) + \sum_{j=0}^{\lceil \log \frac{1}{\epsilon} \rceil} w(J_j\setminus B) = O\left(\frac{\avgH}{\epsilon^3}\log  \frac{1}{\epsilon}\right)w(\mst)
\end{equation*}
which, combined with the reduction to unit-weight $\mst$-edges, proves Theorem~\ref{thm:main} (noting that the stretch condition of $S$ is satisfied since $S$ is a greedy spanner of $G$).

\bigskip 

In the remainder, we prove Lemma~\ref{lem:main-lightness} for a fixed $j \geq 0$. Let $E_i = \Pi_j^i$ for this fixed $j$ and some $i \in \mathbb{Z}^+$.  Let $\ell_i =  \frac{2^{j+1}}{\epsilon^{i}}$; then, the weight of the edges in $E_i$ are in the range $[\ell_i/2,\ell_i)$. Let $E_0 = \mst$.  We refer to the indices $0, 1, 2, \ldots$ of the edge partition as levels.

\subsection{Proof overview}

To prove Lemma~\ref{lem:main-lightness}, we use an amortized analysis, initially assigning each edge of $E_0 = \mst$ a credit of $c = O\left(\frac{\avgH}{\epsilon^3}\right)$.  For each level, we partition the vertices of the spanner into {\em clusters} where each cluster is defined by a subgraph of the graph formed by the edges in levels 0 through $i$. (Note that not every edge of level 0 through $i$ may belong to a cluster; some edges may go between clusters.)  Level $i-1$ clusters are a refinements of level $i$ clusters.  We prove (by induction over the levels), that the clusters for each level satisfy the following \emph{diameter-credit invariants}: 
\begin{description}
\item[DC1] A cluster in level $i$ of diameter $k$ has at least $c \cdot\max\{k,\frac{\ell_i}{2}\}$ credits.
\item[DC2] A cluster in level $i$ has diameter  at most $g\ell_i$ for some constant $g > 2$ (specified later).
\end{description}

We achieve the diameter-credit invariants for the base case (level $0$) as follows.  Although a simpler proof could be given, the following method we use will be revisited in later, more complex, constructions.  Recall that $E_0 = \mst$ and that, in a greedy spanner, the shortest path between endpoints of any edge is the edge itself. If the diameter of $E_0$ is $< \ell_0/2 = O(1)$, edges in the spanner have length at most $\ell_0/2$. Thus, it is trivial to bound the weight of all the spanner edges across all levels using the sparsity of $H$-minor-free graphs.
Assuming a higher diameter, let $\cal T$ be a maximal collection of vertex-disjoint subtrees of $E_0$, each having diameter $\lceil \ell_0/2 \rceil$ (chosen, for example, greedily).  Delete $\cal T$ from $E_0$.  What is left is a set of trees $\cal T'$, each of diameter $< \ell_0/2$.  For each tree $T \in {\cal T}$, let $C_T$ be the union of $T$ with any neighboring trees in $\cal T'$ (connected to $T$ by a single edge of $E_0$).  By construction, $C_T$ has diameter at most $3\ell_0/2 + 1 \le 2\ell_0$ (giving {DC2}).  $C_T$ is assigned the credits of all the edges in the cluster  each of which have credit $c$ (giving {DC1}).

We build the clusters for level $i$ from the clusters of level ${i-1}$ in a series of four phases (Section~\ref{subsec:dc2}).  We call the clusters of level ${i-1}$
{\em $\epsilon$-clusters}, since the diameter of clusters in level ${i-1}$ are an $\epsilon$-fraction of the diameters of clusters in level $i$.  A cluster in level $i$ is induced by a group of $\epsilon$-clusters.  

We try to group the $\epsilon$-clusters so that the diameter of the group is smaller than the sum of the diameters of the $\epsilon$-clusters in the group (Phases 1 to 3).  This {\em diameter reduction} will give us an excess of credit beyond what is needed to maintain {DC1} which allows us to pay for the edges of $E_i$.  We will use the sparsity of $H$-minor free graphs to argue that each $\epsilon$-cluster needs to pay for, on average, a constant number of edges of $E_i$. In Phase 4, we further grow existing clusters via $\mst$ edges and unpaid edges of $E_i$.

Showing that the clusters for level $i$ satisfy invariant {DC2} will be seen directly from the construction. However, satisfying invariant {DC1} is trickier. Consider a path $D$ witnessing the diameter of a level-$i$ cluster $\mathcal{B}$.  Let $\mathcal D$ be the graph obtained from $D$ by contracting $\epsilon$-clusters; we call $\mathcal D$ the {\em cluster-diameter path}. The edges of $\mathcal{D}$ are a subset of $\mst \cup E_i$. If $\mathcal{D}$ does not contain an edge of $E_i$, the credits from the $\epsilon$-clusters and $\mst$ edges of $\mathcal{D}$ are sufficient for satisfying invariant {DC1} for $\mathcal{B}$. However, since edges of $E_i$ are not initialized with any credit, when $\mathcal{D}$ contains an edge of $E_i$, we must use credits of the $\epsilon$-clusters of $\mathcal{B}$ outside $\mathcal{D}$ to satisfy {DC1} as well as pay for $E_i$. Finally, we need to pay for edges of $E_i$ that go between clusters.  We do so in two ways. First, some edges of $E_i$ will be paid for {\em by this level} by using credit leftover after satisfying {DC1}. Second, the remaining edges will be paid for at the end of the entire process (over all levels); we show that there are few such edges over all levels (the edges $B$ of Lemma~\ref{lem:main-lightness}).

In our proof below, the fixed constant $g$ required in DC2 is roughly 100 and $\epsilon$ is sufficiently smaller than $\frac{1}{g}$. For simplicity of presentation, we make no attempt to optimize $g$. We note that a $(1+\epsilon)$-spanner is also a $(1+2\epsilon)$-spanner for any constant $\epsilon$ and the asymptotic dependency of the lightness on $\epsilon$ remains unchanged. That is, requiring that $\epsilon$ is sufficiently small is not a limitation on the range of the parameter $\epsilon$.

\section{Achieving diameter-credit invariants} \label{subsec:dc2}

In this section, we construct clusters for level $i$ that satisfy {DC2} using the induction hypothesis that $\epsilon$-clusters (clusters of level ${i-1}$) satisfy the diameter-credit invariants ({DC1} and {DC2}). Since $\ell_{i-1} = \epsilon \ell_i$, we let $\ell = \ell_i$, and drop the subscript in the remainder.  For {DC2}, we need to group $\epsilon$-clusters into clusters of diameter $\Theta(\ell)$. Let $\mathcal{C}_\epsilon$ be the collection of $\epsilon$-clusters and $\mathcal{C}$ be the set of clusters that we construct for level $i$.  Initially, $\mathcal{C} = \emptyset$. We define a {\em cluster graph} ${\mathcal K}(\mathcal{C}_\epsilon, E_i)$ whose vertices are the $\epsilon$-clusters and edges are the edges of $E_i$.  $\mathcal{K}(\mathcal{C}_\epsilon, E_i)$ can be obtained from the subgraph of $G$ formed by the edges of the $\epsilon$-clusters and $E_i$ by contracting each $\epsilon$-cluster to a single vertex. Recall each $\epsilon$-cluster is a subgraph of the graph formed by the edges in levels 0 through $i-1$.

\begin{observation}\label{obs:K-simple}
$\mathcal{K}(\mathcal{C}_\epsilon, E_i)$ is a simple graph.
\end{observation}
\begin{proof}
  Since $g\epsilon \ell \leq \frac{\ell}{2}$ when $\epsilon$ is sufficiently smaller than $\frac{1}{g}$, there are no self-loops in $\mathcal{K}(\mathcal{C}_\epsilon, E_i)$. Suppose that there are parallel edges $x_1y_1$ and $x_2y_2$ where $x_1, x_2 \in X \in \mathcal{C}_\epsilon$ and $y_1,y_2 \in Y \in \mathcal{C}_\epsilon$.  Let $w(x_2y_2) \le w(x_1y_1)$, w.l.o.g..  Then, the path $P$ consisting of the shortest $x_1$-to-$x_2$ path in $X$, edge $x_2y_2$ and the shortest $y_1$-to-$y_2$ path in $Y$ has length at most $w(x_2y_2) + 2g\epsilon \ell$ by {DC2}.  Since $w(x_2y_2) \le w(x_1y_1)$ and $w(x_1y_1) \ge \ell/2$, $P$ has length at most $(1+4g\epsilon)w(x_1y_1)$.  Therefore, if our spanner is a greedy $(1+4g\epsilon)$-spanner, $x_1y_1$ would not be added to the spanner.
\end{proof}

We call an $\epsilon$-cluster $X$ \emph{high-degree} if its degree in the cluster graph is at least $\frac{20}{\epsilon}$, and \emph{low-degree} otherwise. For each $\epsilon$-cluster $X$, we use $\mathcal{C}(X)$ to denote the cluster in $\mathcal{C}$ that contains $X$. To both maintaining diameter-credit invariants and buying edges of $E_i$, we use credits of $\epsilon$-clusters in $\mathcal{C}(X)$ and $\mst$ edges connecting $\epsilon$-clusters in $\mathcal{C}(X)$. We {\em save} credits of a subset $\mathcal{S}(X)$ of $\epsilon$-clusters of $\mathcal{C}(X)$ and $\mst$-edges connecting $\epsilon$-clusters in $\mathcal{S}(X)$ for maintaining invariant DC1. We then {\em reserve} credits of another subset $\mathcal{R}(X)$ of $\epsilon$-clusters to pay for edges of of $E_i$ incident to $\epsilon$-clusters in $\mathcal{S}(X) \cup \mathcal{R}(X)$. We let other $\epsilon$-clusters in $\mathcal{C}(X) \setminus (\mathcal{S}(X) \cup \mathcal{R}(X))$ \emph{release} their credits to pay for their incident edges of $E_i$; we call such $\epsilon$-clusters \emph{releasing} $\epsilon$-clusters. We designate an $\epsilon$-cluster in $\mathcal{C}(X)$ to be its \emph{center} and let the center collect the credits of $\epsilon$-clusters in $\mathcal{R}(X)$. The credits collected by the center are used to pay for edges of $E_i$ incident to non-releasing $\epsilon$-clusters. 

\subsection{Phase 1: High-degree $\epsilon$-clusters} In this phase, we group high-degree $\epsilon$-clusters. The goal is to ensure that any edge  of $E_i$ not incident to a low-degree $\epsilon$-cluster has both endpoints in the new clusters formed (possibly in distinct clusters). Then we can use sparsity of the subgraph of $\mathcal{K}(\mathcal{C}_\epsilon, E_i)$ induced by the $\epsilon$-clusters that were clustered to argue that the clusters can pay for all such edges; this is possible since this subgraph is a minor of $G$. The remaining edges that have not been paid for are all incident to low-degree $\epsilon$-clusters which we deal with in later phases.

With all $\epsilon$-clusters initially {\em unmarked}, we apply Step 1 until it no longer applies and then apply Step 2 to all remaining high-degree $\epsilon$-clusters at once  and breaking ties arbitrarily:
\begin{enumerate}[Step 1]
\item If there is a high-degree $\epsilon$-cluster $X$ such that all of its neighbor $\epsilon$-clusters in $\mathcal{K}$ are unmarked, we group $X$, edges in $E_i$ incident to $X$ and its neighboring $\epsilon$-cluster into a new cluster $\mathcal{C}(X)$. We then mark all $\epsilon$-clusters in $\mathcal{C}(X)$. We call $X$ the \emph{center} $\epsilon$-cluster of $\mathcal{C}(X)$.

\item After Step 1, any unmarked high-degree $\epsilon$-cluster, say $Y$, must have at least one marked neighboring $\epsilon$-cluster, say $Z$.  We add $Y$ and the edge of $E_i$ between $Y$ and $Z$ to $\mathcal{C}(Z)$ and mark $Y$.
\end{enumerate}

\begin{figure}
\centering
\includegraphics[scale = 1]{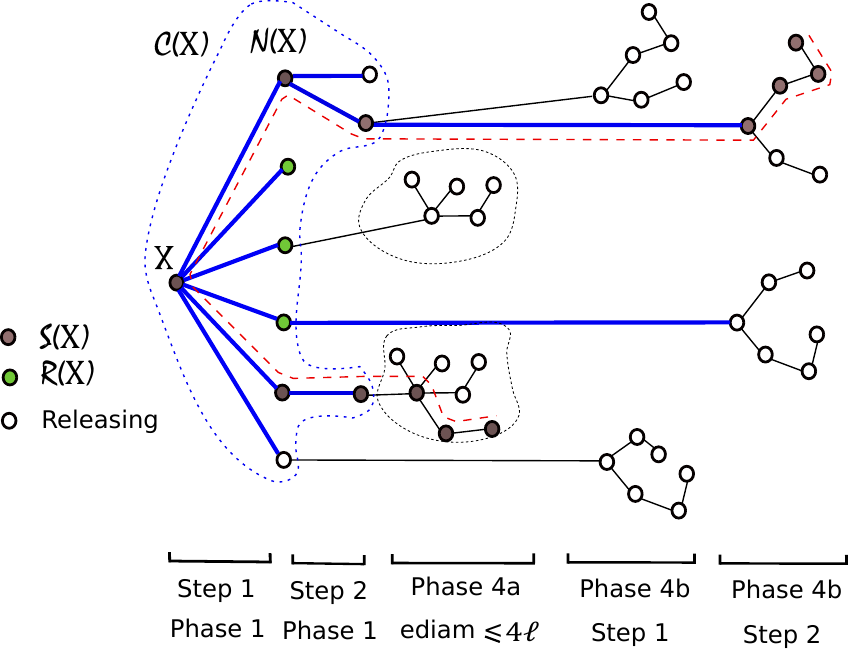}
\caption{A cluster $\mathcal{C}(X)$ formed in Phase 1 is enclosed in the dotted blue curve. The set $\mathcal{S}(X)$ consists of five gray $\epsilon$-clusters inside the dotted blue closed curve and $\mathcal{S}(X)$ consists of green-shaded $\epsilon$-clusters. Remaining hollow $\epsilon$-clusters are releasing. Cluster $\mathcal{C}(X)$ will be augmented further in Phase 4 and augmenting $\epsilon$-clusters are outside the dotted blue curve.  Solid blue edges are in $E_i$ and thin black edges are in $\MST$. The diameter path $\mathcal{D}$  is highlighted by the dashed red curve and $\epsilon$-clusters in $\mathcal{D}$ are gray-shaded.}
\label{fig:P1N}
\end{figure}

In the following, the upper bound is used to guarantee DC2 and the lower bound will be used to guarantee DC1. 
\begin{claim} \label{clm:diam-clsuter-P1}
 The diameter of each cluster added in Phase 1 is at least $\ell$ and at most $(4 + 5g\epsilon)\ell$.
\end{claim}
\begin{proof}
Since the clusters formed are trees each containing at least two edges of $E_i$ and since each edge of $E_i$ has weight at least $\ell/2$, the resulting clusters have diameter at least $\ell$.

Consider an $\epsilon$-cluster $X$ that is the center of a cluster $C$ in Step 1 that is augmented to $\widehat C$ in Step 2 (where, possibly $C = \widehat C$). The upper bound on the diameter of $\widehat C$ comes from observing that any two vertices in $\widehat C$ are connected via at most 5 $\epsilon$-clusters and via at most 4 edges of $E_i$ (each $\epsilon$-cluster that is clustered in Step 2 is the neighbor of a marked $\epsilon$-cluster from Step 1).  Since $\epsilon$-clusters have diameter at most $g\epsilon\ell$ and edges of $E_i$ have weight at most $\ell$, the diameter of $\widehat C$ is at most $(4 + 5g\epsilon)\ell$.
\end{proof}

Let $\mathcal{C}(X)$ be a cluster in Phase 1 with the center $X$. Let $\mathcal{N}(X)$ be the set of $X$'s neighbors in the cluster graph $K(\mathcal{C}_\epsilon, E_i)$.  By construction, $\mathcal{C}(X)$ is a tree of $\epsilon$-clusters. Thus, at most five $\epsilon$-clusters in $\mathcal{C}(X)$  would be in the cluster-diameter path $\mathcal{D}$ while at most three of them are in $\mathcal{N}(X)\cup \{X\}$. We use the credit of $X$ and of two $\epsilon$-clusters in $\mathcal{N}(X)$ for maintaining DC1. Let this set of three $\epsilon$-clusters be $\mathcal{S}(X)$. Since $X$ is high-degree and $\epsilon < 1$, $\mathcal{N}(X)\setminus \mathcal{S}(X)$ has at least $\frac{20}{\epsilon}-2 \geq \frac{18}{\epsilon}$ $\epsilon$-clusters. 
Let $\mathcal{R}(X)$ be any subset of $\frac{18}{\epsilon}$ $\epsilon$-clusters in $\mathcal{N}(X) \setminus \mathcal{S}(X)$. The center $X$ collects the credits of $\epsilon$-clusters in $\mathcal{R}(X)$. We let other $\epsilon$-clusters in $\mathcal{C}\setminus (\mathcal{R}(X) \cup \mathcal{S}(X))$ release their own credits; we call such $\epsilon$-clusters \emph{releasing $\epsilon$-clusters}. By diameter-credit invariants for level ${i-1}$, each $\epsilon$-cluster has at  least $\frac{c\epsilon\ell}{2}$ credits. Thus, we have:

\begin{observation} \label{obs:P1-center-credit}
The center $X$ of  $\mathcal{C}(X)$ collects at least $9c\ell$ credits.
\end{observation}

Let $A_1$ be the set of edges of $E_i$ that have both endpoints in marked $\epsilon$-clusters. 
\begin{claim} \label{clm:paid-credit-P1}
 If $c = \Omega\left( \frac{\avgH}{\epsilon}\right)$, we can buy  edges of $A_1$ using  $c\ell$ credits deposited in the centers and credit of releasing $\epsilon$-clusters.
\end{claim}
\begin{proof} Since the subgraph of $\mathcal{K}$ induced by marked $\epsilon$-clusters and edges of $A_1$ is $H$-minor-free,  each marked $\epsilon$-cluster, on average, is incident to at most $O(\avgH)$ edges of $A_1$. Thus, each $\epsilon$-cluster  must be responsible for buying $\Omega(\avgH) $ edges of $A_1$.

Consider a cluster $\mathcal{C}(X)$. The total credits of each releasing $\epsilon$-clusters is at least $\frac{c\epsilon \ell}{2}$, which is $\Omega(\avgH) \ell$ when $c = \Omega( \frac{\avgH}{\epsilon})$. For non-releasing $\epsilon$-clusters, we use  $c\ell$ credits from their center $X$ to pay for incident edges of $A_1$. Recall that non-releasing $\epsilon$-clusters are in $\mathcal{R}(X) \cup \mathcal{S}(X)$ and:
\begin{equation}\label{eq:Rx-Sx-P1}
|\mathcal{R}(X) \cup \mathcal{S}(X)| \leq 5 + \frac{18}{\epsilon}
\end{equation}

Thus, non-releasing $\epsilon$-cluster are responsible for paying at most $O(\frac{\avgH}{\epsilon})$ edges of $A_1$ and $c\ell$ credits suffice if $c = \Omega(\frac{\avgH}{\epsilon})$.
\end{proof}

By Claim~\ref{clm:paid-credit-P1}, each center $\epsilon$-cluster has at least $8c\ell$ credits remaining after paying for $A_1$. We note that clusters in Phase 1 could be augmented further in Phase 4. We will use these remaining credits at the centers  to pay for edges of $E_i$ in Phase 4.

\subsection{Phase 2: Low-degree, branching $\epsilon$-clusters}
Let $\mf$ be a maximal forest whose nodes are the $\epsilon$-clusters that remain unmarked after Phase 1 and whose edges are $\mst$ edges between pairs of such $\epsilon$-clusters.

Let $\dm(\mathcal{P})$ be the diameter of a path $\mathcal{P}$ in $\mf$, which is the diameter of the subgraph of $G$ formed by edges inside  $\epsilon$-clusters and MST edges connecting $\epsilon$-clusters  of $\mathcal{P}$. We define the \emph{effective diameter} $\edm(\mathcal{P})$ to be the sum of the diameters of the $\epsilon$-clusters  in $\mathcal{P}$. Since the edges of $\mf$ have unit weight (since they are $\mst$ edges), the true diameter of a path in $\mf$ is bounded by the effective diameter of $\mathcal{P}$ plus the number of $\mst$ edges in the path. Since each $\epsilon$-cluster has diameter at least $1$ (by construction of the base case), we have:
\begin{observation} \label{obs:effdiam-vs-diam}
$\dm(\mathcal{P}) \leq 2\edm(\mathcal{P})$.
\end{observation}

We define the effective diameter of a tree (in $\mf$) to be the maximum effective diameter over all paths of the tree.  Let $\mathcal T$ be a tree in $\mf$ that is not a path and such that $\edm(\mathcal{T}) \ge 2\ell$.  Let $X$ be a \emph{branching vertex} of $\mathcal T$, i.e., a vertex of $\mathcal T$ of degree is at least $3$, and let ${\mathcal C}(X)$ be a minimal subtree of $\mathcal T$ that contains $X$ and $X$'s neighbors and such that $\edm(\mathcal{C}(X)) \ge 2\ell$.  We add $\mathcal{C}(X)$ to $\mathcal C$ and delete $\mathcal{C}(X)$ from $\mathcal T$; this process is repeated until no such tree exists  in $\mf$. We refer to $X$ as the \emph{center} $\epsilon$-cluster of $\mathcal{C}(X)$.

\begin{figure}
\centering
\includegraphics[scale = 1]{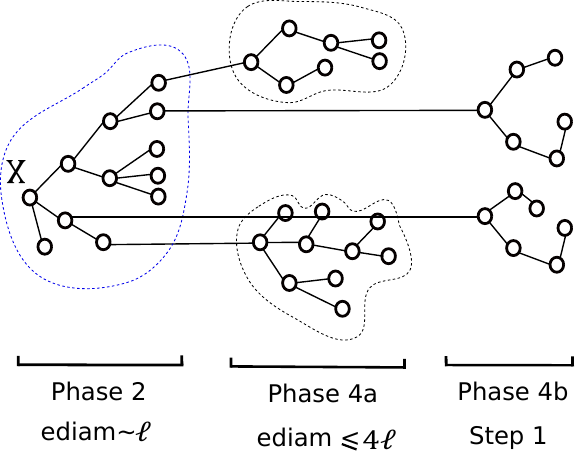}
\caption{A cluster $\mathcal{C}(X)$ formed in Phase 2 is enclosed in the dotted blue curve. $\mathcal{C}(X)$ will be augmented further in Phase 4 and augmenting $\epsilon$-clusters are outside the dotted blue curve. Edges connecting $\epsilon$-clusters are $\MST$ edges.}
\label{fig:P2N}
\end{figure}
\begin{claim} \label{clm:diam-clsuter-P2}
 The diameter of each cluster added  in \emph{Phase 2} is at most $(4+2g\epsilon)\ell $.
\end{claim}
\begin{proof} 
Since $\mathcal{C}(X)$ is minimal, its effective diameter is at most $2\ell + g \epsilon \ell$. The claim follows from Observation~\ref{obs:effdiam-vs-diam}.
\end{proof}

Let $\mathcal{X}$ be a set of $\epsilon$-clusters. We define a subset of $\mathcal{X}$ as follows:
\[
 \trunc{\mathcal{X}}{\rfrac{2g}{\epsilon}} = 
  \begin{cases} 
   \mathcal{X} & \text{if } |\mathcal{X}| \leq \rfrac{2g}{\epsilon}\\
   \text{any subset of }  \rfrac{2g}{\epsilon} \text{ of } \mathcal{X}      & \text{otherwise }
  \end{cases}
\]
By definition, we have:
\begin{equation}\label{eq:size-trunc}
|\trunc{\mathcal{X}}{\rfrac{2g}{\epsilon}}| \leq \frac{2g}{\epsilon}
\end{equation}

 Let $\mathcal{S}(X) = \trunc{\mathcal{C}(X)\cap \mathcal{D}}{\rfrac{2g}{\epsilon}}$ where $\mathcal{D}$ is the diameter path of $\mathcal{C}(X)$. We save credits of $\epsilon$-clusters in $\mathcal{S}(X)$ for maintaining DC1 and we use credits of $\epsilon$-clusters in $\mathcal{C}\setminus \mathcal{S}(X)$ to buy edges of $E_i$ incident to $\epsilon$-clusters in $\mathcal{C}(X)$.  Since $X$ is branching, at least one neighbor $\epsilon$-cluster of $X$, say $Y$, is not in $\mathcal{S}(X)$. Let $\mathcal{R}(X) = \{Y\}$. The center collects credits of clusters in $\mathcal{R}(X)$; other $\epsilon$-clusters in $\mathcal{C}(X)\setminus \{\mathcal{S}(X) \cup \mathcal{R}(X)\}$ release their credits.  

 Let $A_2$ be the set of unpaid edges of $E_i$ incident to $\epsilon$-clusters grouped in Phase 2.

\begin{claim} \label{clm:paid-credit-P2}
 If $c = \Omega( \frac{g}{\epsilon^3})$, we can buy edges  of $A_2$ using $\frac{c\epsilon \ell}{6}$ credits from the center $\epsilon$-clusters and half the credit from releasing $\epsilon$-clusters.
\end{claim}
\begin{proof}
Consider a cluster $\mathcal{C}(X)$ formed in Phase 2. Recall $\epsilon$-clusters in Phase 2 are low-degree. Thus, each $\epsilon$-cluster in $\mathcal{C}(X)$  is incident to at most $\frac{20}{\epsilon} $ edges of $A_2$. We need to argue that each $\epsilon$-cluster has at least  $\frac{20 \ell}{\epsilon} = \Omega(\frac{\ell}{\epsilon}) $ credits to pay for edges of $A_2$.  By invariant DC1 for level ${i-1}$, half credits of releasing $\epsilon$-clusters are at least $\frac{c\epsilon \ell}{4}$, which is  $\Omega(\frac{1}{\epsilon}) \ell$ when $c = \Omega( \frac{1}{\epsilon^2})$.

Since $|\mathcal{R}(X)| = 1$, the center $X$ collects at least $\frac{c\epsilon \ell}{2}$ credits by invariant DC1 for level ${i-1}$. Recall non-releasing $\epsilon$-clusters are all in $\mathcal{S}(X)$. Thus, by Equation~\ref{eq:size-trunc}, the total number of edges of $A_2$ incident to $\epsilon$-clusters in $\mathcal{S}(X) \cup \mathcal{R}(X)$ is at most:
\begin{equation*}
\left(\frac{2g}{\epsilon} + 1\right)\frac{20}{\epsilon} = O\left(\frac{g}{\epsilon^2}\right)
\end{equation*}

Since $c = \Omega(\frac{g}{\epsilon^3})$, $\frac{c\epsilon \ell}{6}$ credits of the center $X$ is at least $\Omega(\frac{g\ell}{\epsilon^2})$ which suffices to buy all edges of $A_2$ incident to $\epsilon$-clusters in $\mathcal{S}(X) \cup \mathcal{R}(X)$.
\end{proof}

We use remaining half the credit of releasing $\epsilon$-clusters to achieve invariant DC1. More details will be given later when we show diameter-credit invariants of $\mathcal{C}(X)$.

\subsection{Phase 3: Grouping $\epsilon$-clusters in high-diameter paths} 

  In this phase, we consider components of $\mf$ that are paths with high effective diameter.  To that end, we partition the components of $\mf$ into HD-components (equiv.\ HD-paths), 
 those with (high) effective diameter at least $4\ell$ (which are all paths) and LD-components, 
 those with (low) effective diameter less that $4\ell$ (which may be paths or trees).  

\subsubsection*{Phase 3a: Edges of $E_i$ within an HD-path} Consider an HD-path $\mathcal P$  that has an edge $e\in E_i$ with endpoints in $\epsilon$-clusters $X$ and $Y$ of $\mathcal{P}$ such that the two disjoint affices ending at $X$ and $Y$ both have effective diameter at least $2\ell$. We choose $e$ such that there is no other edge with the same property on the $X$-to-$Y$ subpath of $\mathcal P$ (By Observation~\ref{obs:K-simple}, there is no edge of $E_i$ parallel to $e$). Let ${\mathcal P}_{X,Y}$ be the $X$-to-$Y$ subpath of $\mathcal P$.  By the stretch guarantee of the spanner, $\dm({\mathcal P}_{X,Y}) \geq (1 + s\epsilon)w(e)$.  Let $\mathcal{P}_X$ and $\mathcal{P}_Y$ be \emph{minimal subpaths} of the disjoint affices of $\mathcal{P}$ that end at $X$ and $Y$, respectively, such that the effective diameters of $\mathcal{P}_X$ and $\mathcal{P}_Y$ are at least $2\ell$.  $\mathcal{P}_X$ and $\mathcal{P}_Y$ exist by the way we choose $e$.
 
 \paragraph*{Case 1: \boldmath$\edm(\mathcal{P}_{X,Y})\le2\ell$}
We construct a new cluster consisting of (the $\epsilon$-clusters and $\mst$ edges of) $\mathcal{P}_{X,Y}$, $\mathcal{P}_X$, $\mathcal{P}_Y$ and edge $e$ (see Figure~\ref{fig:phase4a}(a)).  We refer to, w.l.o.g, $X$ as the center $\epsilon$-cluster of the new cluster.

\begin{claim} \label{clm:diam-cluster-P3a}
 The diameter of each cluster added in Case 1 of Phase 3a is at least $\frac{\ell}{2}$ and at most $(12+4\epsilon g)\ell $.
\end{claim}
\begin{proof}
Since the new cluster contains edge $e$ of $E_i$ and, in spanner $S$, the shortest path between endpoints of any edge is the edge itself, we get the lower bound of the claim. The effective diameters of $\mathcal{P}_X$ and $\mathcal{P}_Y$ are each at most $(2 + \epsilon g)\ell $ since they are minimal. By Observation~\ref{obs:effdiam-vs-diam}, we get that the diameter is at most:
\begin{equation*}
2(\edm(\mathcal{P}_X) + \edm(\mathcal{P}_Y) + \edm(\mathcal{P}_{X,Y}))\leq 4(2+\epsilon g)\ell + 4\ell = (12+4\epsilon g)\ell 
\end{equation*}
\end{proof}

\begin{claim}\label{clm:D-simple-path}
Let $x,y$ be any two vertices of $G$ in a cluster $\mathcal{C}(X)$ added in Case 1 of Phase 3a. Let $P_{x,y}$ be the shortest $x$-to-$y$ path in $\mathcal{C}(X)$ as a subgraph of $G$. Let $\mathcal{P}_{x,y}$ be obtained from $P_{x,y}$ by contracting $\epsilon$-clusters  into a single vertex. Then,  $\mathcal{P}_{x,y}$ is a simple path.
\end{claim}
\begin{proof}
By construction, the only cycle of $\epsilon$-clusters  in $\mathcal{C}(X)$ is $\mathcal{P}_{X,Y} \cup \{e\}$ (see Figure~\ref{fig:phase4a}(a)). Therefore, if ${\cal P}_{x,y}$ is not simple, $e \in \mathcal{P}_{x,y}$ and $\mathcal{P}_{x,y}$ must enter and leave $\mathcal{P}_{X,Y}$ at some $\epsilon$-cluster $Z$.  In this case, $\mathcal{D}$ could be short-cut through $Z$, reducing the weight of the path by at least $w(e) \ge \ell/2$ and increasing its weight by at most $\dm(Z) \le g\epsilon \ell$.  This contradicts the shortness of $P_{x,y}$ for $\epsilon$ sufficiently smaller than $\frac{1}{g}$ ($g\epsilon < \frac{1}{2}$).
\end{proof}

\begin{figure}
\centering
\includegraphics[scale = 1]{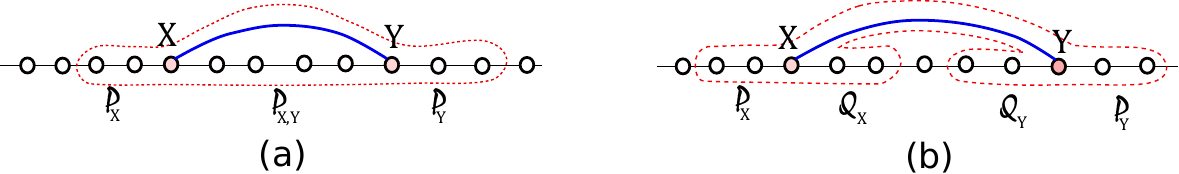}
\caption{(a) A cluster of $\mathcal{C}$ in Case 1 of \textbf{Phase 3a} and (b) a cluster of $\mathcal{C}$ in Case 2 of \textbf{Phase 3a}.  Thin edges are edges of $\MST$, solid blue edges are edges of $E_i$ and vertices are $\epsilon$-clusters. Edges and vertices inside the dashed red curves are grouped into a new cluster.}
\label{fig:phase4a}
\end{figure}

Since $\mathcal{P}_{X,Y} \cup \{e\}$ is the only cycle of $\epsilon$-clusters, by Claim~\ref{clm:D-simple-path}, $\epsilon$-clusters in $\mathcal{D} \cap \mathcal{C}(X)$ form a simple subpath of $\mathcal{D}$ where $\mathcal{D}$ is the diameter path of $\mathcal{C}(X)$. We have:

\begin{observation} \label{obs:Pxy-no-in-D}
   $\mathcal{P}_{X,Y} \not\subseteq \mathcal{D}$.
\end{observation}

\begin{proof}
  For otherwise, $\mathcal{D}$ could be shortcut through $e$ at a cost of
  \begin{eqnarray*}
    & \le&  \underbrace{\dm(X) + \dm(Y) + w(e)}_\text{cost of shortcut}  - \underbrace{(\dm(\mathcal{P}_{X,Y}) - \dm(X) - \dm(Y))}_\text{lower bound on diameter} \\
    & \le&  +w(e) +4g\epsilon\ell -(1+s\epsilon)w(e) \quad \text{(by the stretch condition for }e)\\
    & \le& 4g\epsilon\ell -s\epsilon\ell/2 \quad \text{(since }w(e) \geq \ell/2)
  \end{eqnarray*}
 This change in cost is negative for $s \geq 8g +1$.
\end{proof}

Let $\mathcal{S}(X) = \trunc{ \mathcal{D} \cap \mathcal{C}(X)}{\rfrac{2g}{\epsilon}}$ and $\mathcal{R}(X) = \trunc{ \mathcal{C}(X) \setminus \mathcal{D}}{\rfrac{2g}{\epsilon}}$.  The center $X$ collects the credits of $\epsilon$-clusters in $\mathcal{R}(X)$ and  $\mst$ edges outside $\mathcal{D}$ connecting $\epsilon$-clusters of $\mathcal{C}(X)$. We let other $\epsilon$-clusters in $\mathcal{C}(X)\setminus (\mathcal{R}(X) \cup \mathcal{S}(X))$ release their credits.

\begin{claim}\label{clm:P3a1-credit-center}
If $\mathcal{D}$ does not contain $e$, then $X$ has at least $\frac{c\epsilon\ell}{2}$ credits. Otherwise, $X$ has at least $cw(e) + \frac{c\epsilon\ell}{2}$  credits.
\end{claim}
\begin{proof}
If $\mathcal{C}(X)\setminus \mathcal{D}$ contains at least $\frac{2g}{\epsilon}$ $\epsilon$-clusters, then $|\mathcal{R}(X)| = \frac{2g}{\epsilon}$. Thus, by invariant DC1 for level ${i-1}$, the total credit of $\epsilon$-clusters in $\mathcal{R}(X)$ is at least:
\begin{equation*}
\begin{split}
\frac{2g}{\epsilon} \cdot \frac{c\epsilon \ell}{2} = gc\ell &\geq c\ell + \frac{c\epsilon\ell}{2} \quad \mbox{ (for } g\geq 2 \mbox{ and } \epsilon < 1) \\
&\geq cw(e) + \frac{c\epsilon\ell}{2} \quad \mbox{(since } w(e) \leq \ell)
\end{split}
\end{equation*}
Thus, we can assume that $\mathcal{C}(X)\setminus \mathcal{D}$ contains less than $\frac{2g}{\epsilon}$ $\epsilon$-clusters. In this case, $\mathcal{R}(X) = \mathcal{C}(X)\setminus \mathcal{D}$. Since $\mathcal{P}_{X,Y} \not\subseteq \mathcal{D}$ by Observation~\ref{obs:Pxy-no-in-D}, $\mathcal{D}$ does not contain, w.l.o.g., $\mathcal{P}_X$. Thus,  $\mathcal{R}(X)$ contains at least one $\epsilon$-cluster and the claim holds for the case that $e \not\in \mathcal{D}$. 

Suppose that $\mathcal{D}$ contains $e$ and an internal $\epsilon$-clusters of $\mathcal{P}_{X,Y}$, then w.l.o.g., $\mathcal{D}$ does not contain $\mathcal{P}_X \setminus X$.  $\mathcal{P}_X \setminus X$ has credit $2c\ell - g\epsilon c \ell $. Since $\ell \geq w(e)$ and $\ell - g\epsilon\ell \geq \frac{\epsilon \ell}{2}$ when $\epsilon$ is sufficiently small ($\epsilon \leq \frac{2}{1 + 2g}$), the claim holds.

If $\mathcal{D}$ contains $e$ but no internal $\epsilon$-clusters of $\mathcal{P}_{X,Y}$, then
\begin{equation} \label{eq:dm-L-prime-case2}
\begin{split}
& \dm(\mathcal{P}_{X,Y} \setminus \{X,Y\}) \\
& \ge \dm(\mathcal{P}_{X,Y}) - \dm(X)-\dm(Y) \\
& \ge (1+ s\epsilon)w(e) - \dm(X)-\dm(Y)\ \ \ (\mbox{by the stretch condition}) \\
& \ge w_{e} + s\epsilon\ell/2 - 2g\epsilon\ell\ \ \ (\mbox{by bounds on $w(e)$ and DC2}) \\
& \ge w_e + \epsilon \ell/2 \ \ \ (\mbox{for $s \ge 8g + 1$, as previously required})
\end{split} 
\end{equation}

The credit of the $\mst$ edges and $\epsilon$-clusters of $\mathcal{P}_{X,Y} \setminus \{X,Y\}$ is at least:
\begin{equation}
\begin{split}
&c \cdot(\mst(\mathcal{P}_{X,Y} \setminus \{X,Y\}) + \edm(\mathcal{P}_{X,Y} \setminus \{X,Y\})) \\&\ge c\cdot \dm(\mathcal{P}_{X,Y} \setminus \{X,Y\})\\&\geq c(w_e + \epsilon \ell/2)  \qedhere
\end{split}
\end{equation}
\end{proof} 

Let $A_3$ be the set of \emph{unpaid edges} of $E_i$ incident to $\epsilon$-clusters of clusters in Case 1 of Phase 3a.

\begin{claim}\label{clm:paid-credit-P3a1}
If $ c= \Omega(\frac{g}{\epsilon^3})$,  we can buy edges  of $A_3$ using $\frac{c\epsilon\ell}{6}$ credits from each center and credits of releasing $\epsilon$-clusters.
\end{claim}
\begin{proof}
Consider a cluster $\mathcal{C}(X)$ in Phase 3. Similar to Claim~\ref{clm:paid-credit-P2}, releasing $\epsilon$-clusters can pay for their incident edges in $A_3$ when $c = \Omega(\frac{1}{\epsilon^2})$. By construction, non-releasing clusters of $\mathcal{C}(X)$ are in $\mathcal{S}(X) \cup \mathcal{R}(X)$. Since $|\mathcal{R}(X)| \leq \frac{2g}{\epsilon}$ and $|\mathcal{S}(X)| \leq \frac{2g}{\epsilon}$ by Equation~\eqref{eq:size-trunc} and since clusters now we are considering have low degree, there are at most
\begin{equation*}
\frac{4g}{\epsilon} \cdot \frac{20}{\epsilon}  = O\left( \frac{g}{\epsilon^2}\right)
\end{equation*}
edges of $A_3$ incident to non-releasing $\epsilon$-clusters. Thus, if $ c= \Omega(\frac{g}{\epsilon^3})$,  $\frac{c\epsilon\ell}{6} \geq \Omega(\frac{g}{\epsilon^2})\ell$. That implies  $\frac{c\epsilon\ell}{4}$ credits of $X$ suffice to pay for all edges of $A_3$ incident to non-releasing $\epsilon$-clusters.
\end{proof}

\paragraph*{Case 2: {\boldmath$\edm(\mathcal{P}_{X,Y})>2\ell$}} Refer to Figure~\ref{fig:phase4a}(b). Let $\mathcal{Q}_X$ and $\mathcal{Q}_Y$ be \emph{minimal} affices of $\mathcal{P}_{X,Y}$
such that each has effective diameter at least $\ell$.  We construct a new cluster consisting of (the $\epsilon$-clusters and $\mst$ edges of) $\mathcal{P}_X$, $\mathcal{P}_Y$, $\mathcal{Q}_X$ and $\mathcal{Q}_Y$ and edge $e$. We refer to $X$ as the center of the new cluster.

We apply Case 1 to all edges of $E_i$ satisfying the condition of Case 1 until no such edges exist.  We then apply Case 2 to all remaining edges of $E_i$ satisfying the conditions of Case 2.  After each new cluster is created (by Case 1 or 2), we delete the $\epsilon$-clusters in the new cluster from $\mathcal{P}$, reassign the resulting components of $\mathcal{P}$ to the sets of HD- and LD-components.  At the end, any edge of $E_i$ with both endpoints in the same HD-path have both endpoints in two disjoint affixes of effective diameter less than $2\ell$.

We bound the diameter and credit of the centers of clusters in Case 2 of Phase 3a in Phase 3b. 

\subsubsection*{Phase 3b: Edges of $E_i$ between HD-paths} Let $e$ be an edge of $E_i$ that connects $\epsilon$-cluster $X$ of HD-path $\mathcal{P}$ to $\epsilon$-cluster $Y$ of different HD-path $\mathcal{Q}$ such that none affix of effective diameter less than $2\ell$ of $\mathcal{P}$ contains $X$ and none affix of effective diameter less than $2\ell$ of $\mathcal{Q}$ contains $Y$. Such edge $e$ is said to have  both endpoints \emph{far} from endpoint $\epsilon$-clusters of $\mathcal{P}$ and $\mathcal{Q}$.

Let $\mathcal{P}_X$ and $\mathcal{Q}_X$ be  \emph{minimal} edge-disjoint subpaths of $\mathcal{P}$ that end at $X$ and each having effective diameter at least $2\ell$.  ($\mathcal{P}_X$ and $\mathcal{Q}_X$ exist by the way we choose edge $e$.)  Similarly, define $\mathcal{P}_Y$ and $\mathcal{Q}_Y$.  We construct a new cluster consisting of (the $\epsilon$-clusters and $\mst$ edges of)  $\mathcal{P}_X, \mathcal{P}_Y, \mathcal{Q}_X, \mathcal{Q}_Y$ and edge $e$ (see Figure~\ref{fig:phase4b}). We refer to $X$ as the center of the new cluster. We then delete the $\epsilon$-clusters in the new cluster from $\mathcal{P}$ and $\mathcal{Q}$, reassign the resulting components of $\mathcal{P}$ and $\mathcal{Q}$ to the sets of HD- and LD-components. We continue to create such new clusters until there are no edges of $E_i$ connecting HD-paths with far endpoints. 

We now bound the diameter and credits of the center of a cluster, say $\mathcal{C}(X)$, that is formed in Case 2 of Phase 3a or in Phase 3b. By construction in both cases,  $\mathcal{C}(X)$ consists of two paths  $ \mathcal{P}_X\cup \mathcal{Q}_X$ and $ \mathcal{P}_Y\cap \mathcal{Q}_Y$ connected by an edge $e$.

\begin{claim} \label{clm:diam-cluster-P3b}
 The diameter of each cluster in Case 2 of Phase 3a and Phase 3b is at least $\frac{\ell}{2}$ and at most $(9 + 4\epsilon g)\ell$.
\end{claim}

\begin{proof}
The lower bound follows from the same argument as in the proof of Claim~\ref{clm:diam-cluster-P3a}. Since the effective diameters of $\mathcal{Q}_X$ and $\mathcal{Q}_Y$ are smaller than the effective diameters of $\mathcal{P}_X$ and $\mathcal{P}_Y$, the diameter of the new cluster is bounded by the sum of the diameters of $\mathcal{P}_X$ and $\mathcal{P}_Y$  and $w(e)$.  The upper bound follows from the upper bounds on these diameters as given in the proof of Claim~\ref{clm:diam-cluster-P3a}.
\end{proof}

\begin{figure}
\centering
\includegraphics[scale = 1.5]{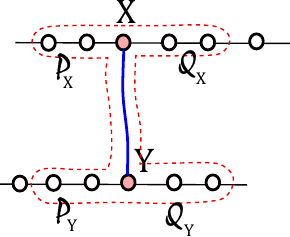}
\caption{A cluster of $\mathcal{C}$ in \textbf{Phase 4a}.  Thin edges are edges of $\MST$, solid blue edges are edges of $E_i$ and vertices are $\epsilon$-clusters. Edges and vertices inside the dashed red curves are grouped into a new cluster.}
\label{fig:phase4b}
\end{figure}

We show how to pay for unpaid edges of $E_i$ incident to $\epsilon$-clusters in Case 2 of Phase 3a and Phase 3b.  W.l.o.g, we refer to  $X$ as the center $\epsilon$-cluster of $\mathcal{C}(X)$. Let $\mathcal{S}(X) = \trunc{\mathcal{D} \cap \mathcal{C}(X)}{\rfrac{2g}{\epsilon}}$ and $\mathcal{R}(X) = \trunc{\mathcal{C}(X) \setminus \mathcal{D}}{\rfrac{2g}{\epsilon}}$ where $\mathcal{D}$ is the cluster-diameter path of $\mathcal{C}(X)$. We save credits of $\epsilon$-clusters in $\mathcal{S}(X)$ for maintaining invariant DC1. The center $X$ collects credits of $\epsilon$-clusters in $\mathcal{R}(X)$. We let other $\epsilon$-clusters in $\mathcal{C}(X)\setminus (\mathcal{S}(X) \cup \mathcal{R}(X))$ to release their credits. 

\begin{claim} \label{clm:P3a2-P3b-credit-center}
The center of a cluster in Case 2 of Phase 3a or in Phase 3b has at least $2c(1-g\epsilon)\ell$ credits.
\end{claim}
\begin{proof}
If $|\mathcal{C}(X)\setminus \mathcal{D} |  \geq \frac{2g}{\epsilon}$, $\mathcal{R}(X)$  has $\frac{2g}{\epsilon}$ $\epsilon$-clusters which have at least $gc\ell$ total credits by invariant DC1 for level ${i-1}$. Since $gc\ell > 2c(1-g\epsilon)\ell$ when $g > 2$, the claim holds. Thus, we assume that $|\mathcal{C}(X)\setminus \mathcal{D} |  < \frac{2g}{\epsilon}$ which implies $\mathcal{R}(X) =\mathcal{C}(X)\setminus \mathcal{D} $. By construction, $\mathcal{S}(X)$ contains $\epsilon$-clusters of at most two of four paths $\mathcal{P}_X,\mathcal{P}_Y,\mathcal{Q}_X,\mathcal{Q}_Y$.  Since each path has effective diameter at least $\ell$, the $\epsilon$-clusters of each path in $\mathcal{R}(X)$ have total diameter at least $\ell - g\epsilon \ell$. By invariant DC1 for level ${i-1}$, each path in  $\mathcal{R}(X)$ has at least $c(1 - g\epsilon)\ell$ credits that implies the claim.
\end{proof}

Let $A_4$ be the set of \emph{unpaid edges} of $E_i$ incident to $\epsilon$-clusters of clusters in Case 2 of Phase 3a and clusters in Phase 3b.

\begin{claim} \label{clm:paid-credit-P3a2-P3b}
If $c = \Omega( \frac{g}{\epsilon^2})$, we can buy edges  of $A_4$ using $c(1-3g\epsilon)\ell$ credits of the centers of clusters in Case 2 of Phase 3a and Phase 3b and credits of releasing $\epsilon$-clusters.
\end{claim}
\begin{proof}Similar to the proof of Claim~\ref{clm:paid-credit-P2}, releasing $\epsilon$-clusters of $\mathcal{C}(X)$ can buy their incident edges in $A_4$ when $c = \Omega(\frac{1}{\epsilon^2})$. By construction, non-releasing $\epsilon$-clusters are in $\mathcal{S}(X) \cup \mathcal{R}(X)$. Since $|\mathcal{R}(X)| \leq \frac{2g}{\epsilon}$ and $|\mathcal{S}(X)| \leq \frac{2g}{\epsilon}$, there are at most $ O(\frac{g}{\epsilon^2})$ edges of $A_4$ incident to non-releasing $\epsilon$-clusters. When $\epsilon$ is sufficiently small ($\epsilon < \frac{1}{6g}$), $c(1 - 3g\epsilon)\ell > \frac{c\ell}{2} $. Thus, if $c= \Omega(\frac{g}{\epsilon^2})$, $\frac{c\ell}{2} =  \Omega(\frac{g}{\epsilon^2})\ell$ and hence,  $c(1 - 3g\epsilon)\ell$ credits suffice to pay for all edges of $A_4$ incident to non-releasing $\epsilon$-clusters of $\mathcal{C}(X)$.
\end{proof}

\subsection{Phase 4: Remaining HD-paths and LD-components}
 
We assume that $\mathcal{C} \not= \emptyset$ after Phase 3. The case when  $\mathcal{C} = \emptyset$ will be handled at the end of this section.

\subsubsection*{Phase 4a: LD-components}  Consider a LD-component $\mathcal{T}$, that has effective diameter less than $4\ell$. By construction, $\mathcal{T}$ must have an $\MST$ edge to a cluster, say $\mathcal{C}(X)$, in $\mathcal{C}$ formed in a previous phase. We include $\mathcal{T}$ and an $\MST$ edge connecting  $\mathcal{T}$ and  $\mathcal{C}(X)$ to  $\mathcal{C}(X)$. Let $A_5$ be the set of unpaid edges of $E_i$ that incident to $\epsilon$-clusters merged into new clusters in this phase. We use credit of the center $X$ and $\epsilon$-clusters in this phase to pay for $A_5$. More details will be given in Phase 4b.

\subsubsection*{Phase 4b: Remaining HD-paths} Let $\mathcal{P}$ be a HD-path. By construction, there is at least one $\mst$ edge connecting $\mathcal{P}$ to an existing cluster in $\mathcal{C}$. Let $e$ be one of them. Greedily break $\mathcal P$ into subpaths such that each subpath has effective diameter at least $2\ell$ and at most $4\ell$. We call a subpath of $\mathcal{P}$ a \emph{long subpath} if it contains at least $\frac{2g}{\epsilon}+1$ $\epsilon$-clusters and \emph{short subpath} otherwise. We process subpaths of $\mathcal P$ in two steps. In Step 1, we process affixes of $\mathcal{P}$, long subpaths of $\mathcal{P}$ and the subpath of $\mathcal{P}$ containing an endpoint $\epsilon$-cluster of $e$. In Step 2, we process remaining subpaths of $\mathcal{P}$. 
 
\paragraph*{Step 1}     If a subpath $\mathcal{P}'$ of $\mathcal{P}$ contain an $\epsilon$-cluster that is incident to $e$, we merge $\mathcal{P}'$ to the cluster in $\mathcal{C}$ that contains another endpoint $\epsilon$-cluster of $e$. We call $\mathcal{P}'$ the \emph{augmenting subpath} of $\mathcal{P}$. We form a new cluster from each long subpath of $\mathcal{P}$ and each affix of $\mathcal{P}$. It could be that one of two affixes of $\mathcal{P}$ is augmenting. We repeatedly apply Step 1 for all HD-paths. The remaining cluster paths which are short subpaths of HD-paths would be handled in Step 2.  We then pay for every unpaid edges of $E_i$ incident to $\epsilon$-clusters in this step. We call a cluster \emph{a long cluster} if it is a long subpath of $\mathcal{P}$ and a \emph{short cluster} if it is a short subpath of $\mathcal{P}$.  

Let $A_6$ be the set of unpaid edges of $E_i$ incident to $\epsilon$-clusters of long clusters.  We show below that each long cluster can both maintain diameter-credit invariant and pay for its incident edges in $A_6$ using credits of its $\epsilon$-clusters. 

Let $A_7$ be the set of unpaid edges of $E_i$ incident to remaining  $\epsilon$-clusters involved in this step; those belong to augmenting subpaths and short affices of HD-paths. We can pay for edges of $A_7$ incident to $\epsilon$-clusters in augmenting subpaths using the similar argument in previous phases. However, we must be careful when paying for other edges of $A_7$ that are incident to $\epsilon$-clusters in short affices of $\mathcal{P}$. Since short affices of $\mathcal{P}$ spend all credits of their children $\epsilon$-clusters to maintain invariant DC1, we need to use credits of $\epsilon$-clusters in $\mathcal{P}'$ to pay for edges of $A_7$ incident to short affices of $\mathcal{P}$.

\paragraph*{Step 2}  Let $\mathcal{P}'$ be a short subpath of $\mathcal{P}$. If edges of $E_i$ incident to $\epsilon$-clusters of  $\mathcal{P}'$ are all paid, we let $\mathcal{P}'$ become a new cluster. Suppose that  $\epsilon$-clusters in  $\mathcal{P}'$ are incident to at least one unpaid edge of $E_i$, say $e$. We have:

\begin{observation} \label{obs:unpaid-edges-P4b2} Edge $e$ must be incident to an $\epsilon$-cluster merged in Phase 1. 
\end{observation}
\begin{proof}
Recall that edges of $E_i$ incident to $\epsilon$-clusters of clusters initially formed in  previous phases except Phase 1 are in $A_2 \cup \ldots \cup A_7$; thus, they are all paid. By construction, edges of $E_i$ between two $\epsilon$-clusters in the same cluster initially formed in Phase 1 are in $A_1 \cup A_5 \cup A_7$ which are also paid. Since $\mathcal{P}'$ is not an affix of $\mathcal{P}$, there is no unpaid edge between two $\epsilon$-clusters of $\mathcal{P}'$ since otherwise $\mathcal{P}'$ would become a new cluster in Phase 3a; that implies the observation.
\end{proof}

We merge $\epsilon$-clusters, $\mst$ edges of $\mathcal{P}'$ and $e$ to the cluster in $\mathcal{C}$ that contains another endpoint of $e$. This completes the clustering process. Let $A_8$ be the set of remaining unpaid edges of $E_i$ incident to  $\epsilon$-clusters involved in Step 2.

We now analyze clusters of $\mathcal{C}$ which are formed or modified in Phase 4. 

\begin{claim} \label{clm:dm-crd-short-affix-clustesr}  Let $\mathcal{B}$ be a short cluster. Then, $\dm(\mathcal{B}) \leq 8\ell$ and credits of $\epsilon$-clusters and $\mst$ edges connecting  $\epsilon$-clusters in $\mathcal{B}$ suffice to maintain invariant DC1 for $\mathcal{B}$.
\end{claim}
\begin{proof}
Since $\edm(\mathcal{B}) \leq 4\ell$, by Observation~\ref{obs:effdiam-vs-diam}, $\dm(\mathcal{B}) \leq 8\ell$. The total credit of $\epsilon$-clusters and $\mst$ edges in $\mathcal{B}$ is at least:
\begin{equation*}
c(|\MST({\cal B})| +  \edm({\cal B})) \geq c \cdot\dm({\cal B}))
\end{equation*}
Since $\edm({\cal B})) \geq 2\ell$, $\mathcal{B}$ has at least $2c\ell$ credits. Thus, $\cal B$ has at least $c\cdot\max(\dm(\mathcal{B}), \ell/2)$ credits.
\end{proof}

We show how to pay for edges of $E_i$ in $A_6$ and maintain diameter-credit invariants of long clusters. We use $cr(\mathcal{X})$ to denote the total credit of $\epsilon$-clusters of a set of $\epsilon$-clusters $\mathcal{X}$.

\begin{claim} \label{clm:paid-dm-crd-long-clustesr} Let $\mathcal{B}$ be a long cluster. If $c = \Omega(\frac{g}{\epsilon^3})$ and $g \geq 8$, we can maintain diameter-credit invariants of $\mathcal{B}$ and pay for edges in $A_6$ incident to $\epsilon$-clusters in $\mathcal{B}$ using credits of $\epsilon$-clusters in $\mathcal{B}$. 
\end{claim}
\begin{proof} By construction, $\mathcal{B}$ has effective diameter at most $4\ell$. By Observation~\ref{obs:effdiam-vs-diam}, $\mathcal{B}$ has diameter at most $8\ell$. Thus, $\mathcal{B}$ satisfies invariant DC2 if $g \geq 8$. 
Since $\mathcal{B}$ is a long cluster, it has at least $\frac{2g}{\epsilon} + 1$ $\epsilon$-clusters. Let  $\mathcal{S}$  be a set of $\frac{2g}{\epsilon}$ $\epsilon$-clusters in $\mathcal{B}$ and $X$ be an $\epsilon$-cluster in $\mathcal{B}\setminus \mathcal{S}$. Let $\mathcal{R} = \{X\} $. We save credits of $\mathcal{S}$ for maintaining invariant DC1 of $\mathcal{B}$ and use credits of $\mathcal{R}$ to pay for edges of $A_6$ incident to $\epsilon$-clusters in $\mathcal{S} \cup \mathcal{R}$. Since $|\mathcal{S} \cup \mathcal{R}| = \frac{2g}{\epsilon} + 1$ and $\epsilon$-clusters in $\mathcal{S} \cup \mathcal{R}$ are low-degree, there are at most $O(\frac{g}{\epsilon^2})$ edges of $A_6$ incident to $\epsilon$-clusters in $\mathcal{S} \cup \mathcal{R}$. By invariant DC1 for level ${i-1}$, $\mathcal{R}$ has at least $\frac{c\epsilon \ell}{2}$ credits which is sufficient to pay for $O(\frac{g}{\epsilon^2})$ edges of $A_6$ when $c = \Omega(\frac{g}{\epsilon^3})$.  We let other $\epsilon$-clusters in $\mathcal{B}\setminus (\mathcal{S}\cup \mathcal{R})$ pay for their incident edges of $A_6$ using their credits. This is sufficient when $c = \Omega(\frac{1}{\epsilon^2})$ since each $\epsilon$-cluster is incident to at most $\frac{20}{\epsilon}$ edges and has at least $\frac{c\epsilon \ell}{2}$ credits. 

 We use credits of $\mathcal{S}$ to maintain invariant DC1. Since $|\mathcal{S}| = \frac{2g}{\epsilon}$ and each $\epsilon$-clusters has at least $\frac{c\epsilon\ell}{2}$ credits, $cr(\mathcal{S}) \geq g\ell$.  Since $\dm(\mathcal{B}) \leq g\ell$ by DC2, $cr(\mathcal{S}) \geq c\dm(\mathcal{B})$. Thus, $cr(\mathcal{S}) \geq c\dot \max(\dm(\mathcal{B}), \ell/2)$; invariant DC1 is satisfied. 
\end{proof}

Let $\mathcal{C}(X)$ be a cluster in $\mathcal{C}$ before Phase 4. Let $\mathcal{C}'(X), \mathcal{C}''(X)$ and $\mathcal{C}'''(X)$ be the corresponding clusters that are augmented from $\mathcal{C}(X)$ in Phase 4a, Step 1 of Phase 4b and Step 2 of Phase 4b, respectively. It could be that any two of three clusters are the same. 

By construction in Phase 4a, LD-components are attached to $\mathcal{C}(X)$ via $\mst$ edges. Recall each LD-component has effective diameter at most $4\ell$ and hence, diameter at most $8\ell$ by Observation~\ref{obs:effdiam-vs-diam}. Thus, $\dm(\mathcal{C}'(X)) - \dm(\mathcal{C}(X)) \leq 16\ell+2$.  By construction in Step 1 of Phase 4b, subpaths of effective diameter at most $4b\ell$ are attached to $\mathcal{C}'(X)$ via $\mst$ edges. Thus, $\dm(\mathcal{C}''(X)) -\dm(\mathcal{C}'(X)) \leq 16\ell + 2$. We have:
\begin{claim} \label{clm:diam-cluster-P4a-and-P4b1}
$\dm(\mathcal{C}''(X)) - \dm(\mathcal{C}(X)) \leq 32\ell + 4$.
\end{claim}

By construction in Step 2 of Phase 4b, subpaths of HD-paths are attached to $\mathcal{C}''(X)$ via edges of $E_i$.  Since attached subpaths have effective diameter at most $4\ell$, by Observation~\ref{obs:effdiam-vs-diam}, we have:
\begin{claim}\label{clm:diam-P4b2}
$\dm(\mathcal{C}'''(X)) - \dm(\mathcal{C}''(X)) \leq 18\ell$
\end{claim}

We are now ready to show invariant DC2 for $\mathcal{C}'''(X)$.
\begin{claim}\label{clm:inv-DC2-CX}
$\dm(\mathcal{C}'''(X)) \leq g\ell$ when $g \geq 70$.
\end{claim}
\begin{proof}
From Claim~\ref{clm:diam-clsuter-P1}, Claim~\ref{clm:diam-clsuter-P2}, Claim~\ref{clm:diam-cluster-P3a} and Claim~\ref{clm:diam-cluster-P3b}, $\mathcal{C}(X)$ has diameter at most:
\begin{equation*}
\max((4+5g\epsilon)\ell,(4+2g\epsilon)\ell, (12+4g\epsilon)\ell, (9+4g\epsilon)\ell) = (12 + 4g\epsilon)\ell 
\end{equation*} 
which is at most $16\ell$ when $\epsilon$ is sufficiently small ($\epsilon < 1/g$). By Claim~\ref{clm:diam-cluster-P4a-and-P4b1} and Claim~\ref{clm:diam-P4b2}, $\mathcal{C}'''(X)$ has diameter at most:
\begin{equation*}
16\ell + 32\ell + 4 + 18\ell = 66\ell + 4 \leq 70\ell
\end{equation*}
since $\ell \geq 1$.
\end{proof}

Recall we show how to pay for edges in $A_1,A_2,A_3,A_4,A_6$ before. It remains to show how to pay for edges in $A_5 \cup A_7 \cup A_8$. We first consider edges in $A_5 \cup A_7$. Recall $\mathcal{S}(X)= \trunc{(\mathcal{D} \cap \mathcal{C}(X))}{\rfrac{2g}{\epsilon}}$ where $\mathcal{D}$ is the cluster-diameter path. We call $\epsilon$-clusters in $\mathcal{C}''(X)\setminus \mathcal{C}(X)$ \emph{augmenting} $\epsilon$-clusters. Let $\mathcal{S}''(X) = \trunc{(\mathcal{D} \cap (\mathcal{C}''(X)\setminus \mathcal{C}(X))}{\rfrac{2g}{\epsilon}}$ be the set of augmenting $\epsilon$-clusters that are in the diameter path $\mathcal{D}$. We save credits of $\epsilon$-clusters in $\mathcal{S}''(X)$ for maintaining DC1 and let other augmenting $\epsilon$-clusters release their credits. 

\begin{claim} \label{clm:paid-credit-P4a-P4b1}
If $c = \Omega(\frac{g}{\epsilon^2})$, we can buy edges in $A_5 \cup A_7$ using $\frac{c\epsilon \ell}{3}$ credits of the cluster centers and credits of releasing augmenting $\epsilon$-clusters.
\end{claim}
\begin{proof}
We use $\frac{c\ell}{6}$ credits of $X$ to pay for edges of $A_5\cup A_7$ incident to $\epsilon$-clusters in $\mathcal{S}''(X)$.  Recall each $\epsilon$-cluster is incident to at most $\frac{20}{\epsilon}$ edges of $E_i$ since it is low-degree. Thus, $\epsilon$-clusters in $\mathcal{S}''(X)$ are incident to  most $O(\frac{g}{\epsilon^2})$ edges of $A_5\cup A_7$. Hence, $\frac{c\ell}{6}$  credits suffice when $c = \Omega(\frac{g}{\epsilon^2})$.  We let releasing augmenting $\epsilon$-clusters of LD-components to pay for their incident edges  of $A_5$. This is sufficient when $c = \Omega(\frac{1}{\epsilon^2})$. Thus, all edges of $A_5$ are paid. We now turn to edges of $A_7$.

Let $\mathcal{P}_1, \mathcal{P}_2, \mathcal{P}_3$ be three segments of a HD-path $\mathcal{P}$ in Step 1 where $\mathcal{P}_1, \mathcal{P}_2$ are affixes of $\mathcal{P}$ and $\mathcal{P}_3$ is the augmenting subpath of $\mathcal{P}$. It could be that $\mathcal{P}_1 = \mathcal{P}_3$ or $\mathcal{P}_2 = \mathcal{P}_3$.
 Since edges of $E_i$ incident to long clusters are paid in Claim~\ref{clm:paid-dm-crd-long-clustesr}, $\epsilon$-clusters of $\mathcal{P}_i$, $1 \leq  i \leq 2$, are incident to unpaid edges of $E_i$ only when $\mathcal{P}_i$ is a short cluster and thus, incident to at most $O(\frac{g}{\epsilon^2})$ edges of $A_7$. Note that in Claim~\ref{clm:dm-crd-short-affix-clustesr}, we use all credits of $\epsilon$-clusters and $\mst$ edges of $\mathcal{P}_i$ to maintain diameter-credit invariants and we need to pay for edges of $A_7$ incident to $\mathcal{P}_i$. We consider two cases:

\begin{enumerate}
\item If $\mathcal{P}_3 \cap \mathcal{S}''(X) = \emptyset$, then $\epsilon$-clusters in $\mathcal{P}_3$ are releasing. Recall $\mathcal{P}_3$ has effective diameter at least $2\ell$. We let each $\epsilon$-cluster in $\mathcal{P}_3$ pay for its incident edges of $A_7$ using half of its credits, which is at least $\frac{c\epsilon \ell}{4}$ by invariant DC1 for level ${i-1}$. This amount of credits is enough when $c = \Omega(\frac{1}{\epsilon^2})$. The total remaining credit from $\epsilon$-clusters of $\mathcal{P}_3$ is at least $c\ell$, that is sufficient to pay for $O(\frac{g}{\epsilon^2})$ edges of $A_7$ incident to $\epsilon$-clusters of $\mathcal{P}_1 \cup \mathcal{P}_2$ when $c = \Omega(\frac{g}{\epsilon^2})$.

\item  If $\mathcal{P}_3 \cap \mathcal{S}''(X) \not= \emptyset$, we use $c\ell/6$ credits of the center $X$ of $\mathcal{C}''(X)$ to pay for edges of $A_7$ incident to $\epsilon$-clusters in $\mathcal{X} = (\mathcal{P}_1 \cap \mathcal{S}''(X))\cup \mathcal{P}_2 \cup \mathcal{P}_3$. Recall $\mathcal{S}''(X)$ has at most $\frac{2g}{\epsilon}$ $\epsilon$-clusters, $\mathcal{X}$ has at most $\frac{6g}{\epsilon}$  $\epsilon$-clusters. Thus, $\epsilon$-clusters in $\mathcal{X}$ are incident to at most $\frac{120g}{\epsilon^2}$ edges in $A_7$. Since there are at most two augmenting subpaths that contain $\epsilon$-clusters of the cluster diameter path $\mathcal{D}$, $X$ only need to pay for   at most $\frac{240g}{\epsilon^2} = O(\frac{g}{\epsilon^2})$ edges. Thus, $\frac{c\ell}{6}$ credits are sufficient if $c = \Omega(\frac{g}{\epsilon^2})$.  Other $\epsilon$-clusters of $\mathcal{P}_3\setminus \mathcal{S}''(X)$ are releasing and we can use their released credits to pay for their incident edges of $A_7$.
\end{enumerate}
\end{proof}

We now show how to pay for edges of $A_8$ which consists of edges of $E_i$ incident to $\epsilon$-clusters in Step 2 of Phase 4b. By Observation~\ref{obs:unpaid-edges-P4b2}, $\mathcal{C}(X)$ is formed in Phase 1. Let $\mathcal{S}'''(X)$ be augmenting $\epsilon$-clusters in $\mathcal{D}$ of $\mathcal{C}'''(X)$ that are not in $\mathcal{S}(X) \cup \mathcal{S}''(X)$. We save credit of   $\mathcal{S}'''(X)$ for maintaining DC1 and let other augmenting $\epsilon$-clusters release their credits.

\begin{claim} \label{clm:paid-credit-P4b2} If $c = \Omega(\frac{g}{\epsilon^2})$, we can pay for edges of $A_8$ incident to $\epsilon$-clusters in $\mathcal{C}'''(X)$ using credits of releasing $\epsilon$-clusters and $c\ell$ credits of the center $X$. 
\end{claim}
\begin{proof}

Since augmenting $\epsilon$-clusters are low-degree, each augmenting $\epsilon$-cluster is incident to at most $\frac{20}{\epsilon}$ edges of $A_8$. When $c = \Omega(\frac{1}{\epsilon^2})$, $\frac{c\epsilon\ell}{2}$ credits of each releasing $\epsilon$-cluster suffice to buy their incident edges of $A_8$. 

By construction, the augmenting subpath $\mathcal{P}'$ in Step 1 of Phase 4b is a short path. Since the cluster-diameter path $\mathcal{D}$ contains $\epsilon$-clusters of at most two short subpaths of HD-paths, $|\mathcal{S}'''(X)|\leq \frac{4g}{\epsilon}$. Thus, there are at most $O(\frac{g}{\epsilon^2})$ edges of $A_8$ incident to non-releasing $\epsilon$-clusters. Hence,  $c\ell$ credits of $X$  suffice to pay for such edges when $c = \Omega(\frac{g}{\epsilon^2})$.
\end{proof}

It remains to maintain invariant DC1 for clusters in $\mathcal{C}$. We have:
\begin{claim} \label{clm:inv1-bigS}
If any of the sets $\mathcal{S}(X), \mathcal{S}''(X)$ and $\mathcal{S}'''(X)$ has at least $\frac{2g}{\epsilon}$ $\epsilon$-clusters, then $\mathcal{C}'''(X)$ satisfies invariant DC1. 
\end{claim}
\begin{proof}
Suppose,~w.l.o.g, say $\mathcal{S}(X)$ has at least  $\frac{2g}{\epsilon}$ $\epsilon$-clusters. Then, by DC1 for level ${i-1}$, the total credits of $\epsilon$-clusters in $\mathcal{S}(X)$ is at least:
\begin{equation*}
 \frac{2g}{\epsilon}\cdot \frac{c\epsilon \ell}{2} = gc\ell
\end{equation*}
which is at least $c\cdot\max(\dm(\mathcal{C}'''(X)), \ell/2)$ since $\dm(\mathcal{C}'''(X)) \leq g\ell$ by Claim~\ref{clm:inv-DC2-CX} and $g > 1$. 
\end{proof}

\begin{claim}\label{clm:inv-DC1-CX}
If $c = \Omega(\frac{g}{\epsilon^3})$, we can maintain invariant DC1 of $\mathcal{C}'''(X)$ using credits of $\epsilon$-clusters and $\mst$ edges in $\mathcal{D}$ and the credits of the cluster center $X$.
\end{claim}
\begin{proof}
By Claim~\ref{clm:inv1-bigS}, credits of all $\epsilon$-clusters and $\mst$ edges of $\mathcal{D}$ are saved for maintaining DC1. We prove the claim by case analysis.
\vspace{5mm}

\noindent \textbf{Case 1: $\mathcal{C}(X)$ is formed in Phase 1.} Recall $\mathcal{D}$ contains at most six edges of $E_i$ where four edges of $E_i$ are in $\mathcal{C}(X)$ and two more edges of $E_i$ are by the augmentation in Step 2 of Phase 4b. We use $6c\ell$ credits from $X$ and credits of $\epsilon$-clusters and $\mst$ edges in $\mathcal{D}$. The total credit is:
\begin{equation*}
6c\ell + c(|\MST({\cal D})| +  \edm({\cal D})) \geq c \cdot\dm({\cal D}) = c\cdot \dm(\mathcal{C}'''(X)))
\end{equation*} 
Since $\mathcal{C}'''(X)$ contains an edge in $E_i$, $\dm(\mathcal{C}'''(X)) \geq \ell/2$. Thus, $c\cdot \dm(\mathcal{C}'''(X))) \geq c\ell/2$.

To complete the proof, we need to argue that $X$ has non-negative credits after paying for edges of $E_i$ and maintaining invariant DC1 of $\mathcal{C}'''(X)$. Recall $X$ initially has $9c\ell$ credits by Observation~\ref{obs:P1-center-credit} and loses:
\begin{itemize}
\item $c\ell$ credits in Claim~\ref{clm:paid-credit-P1}.
\item $\frac{c\epsilon \ell}{3}$ credits in Claim~\ref{clm:paid-credit-P4a-P4b1}.
\item $c\ell$ credits to pay for the edges of $A_8$ incident to non-releasing augmenting $\epsilon$-clusters in Step 2 of Phase 4b.
\item $6c\ell$ credits for maintaining DC1 of $\mathcal{C}'''(X)$.
\end{itemize}
Thus, $X$ still has:
\begin{equation*}
9c\ell - 8c\ell - \frac{c\epsilon \ell}{3} = c(1 - \frac{\epsilon}{3})\ell
\end{equation*}
which is non-negative since $\epsilon < 1$. 
\vspace{5mm}

\noindent \textbf{Case 2: $\mathcal{C}(X)$ is formed in Phase 2.} Recall the center $X$ collects at least $\frac{\epsilon \ell}{2}$ from a neighbor $Y$ of $X$ ($\mathcal{R}(X) = \{Y\}$). We observe that credits in $X$ is taken totally by at most $\frac{\epsilon \ell}{2}$ in Claim~\ref{clm:paid-credit-P2} and Claim~\ref{clm:paid-credit-P4a-P4b1}. Thus, the center still has non-negative credits after buying incident edges $E_i$ when $c = \Omega(\frac{g}{\epsilon^3})$. 

Since credits of $\epsilon$-clusters and $\mst$ edges in $\mathcal{D}$ are reserved and $\mathcal{D}$ does not contain any edge of $E_i$, the total reserved credit is:
\begin{equation}
c(|\MST({\cal D})| +  \edm({\cal D})) \geq c \cdot\dm({\cal D}) = c\cdot \dm(\mathcal{C}''(X)))
\end{equation}
It remains to argue that $\mathcal{C}''(X)$ has at least $\frac{c\ell}{2}$ credits. Note that we do not have lower bound on the diameter of $\mathcal{C}''(X)$ as in other cases. Let $\mathcal{X}$ be the set of releasing $\epsilon$-clusters of $\mathcal{C}(X)$ and $cr(\mathcal{X})$ be the total credits of $\epsilon$-clusters in $\mathcal{X}$. Since $\edm(\mathcal{C}(X)) \geq 2\ell$, we have:
\begin{equation} \label{eq:credit-releasing-X-P2}
cr(\mathcal{X}) +  cr(\mathcal{S}(X)) \geq 2c\ell
\end{equation}

Recall half credit of $\mathcal{X}$ is taken in Claim~\ref{clm:paid-credit-P2}. We use the remaining half to guarantee that the credit of $\mathcal{C}''(X)$ is at least $c\ell/2$.
\vspace{5mm}

\noindent \textbf{Case 3: $\mathcal{C}(X)$ is formed in Case 1 of Phase 3a.}  Recall (in Claim~\ref{clm:P3a1-credit-center}) the center $X$ collects at least $\frac{\epsilon \ell}{2}$ credits if $\mathcal{D}$ does not contain $e$ (we are using notation in Case 1 Phase 3a) and at least  $cw(e) + \frac{\epsilon \ell}{2}$ credits if $\mathcal{D}$ contains $e$. We observe that credits in $X$ is taken totally by at most $\frac{\epsilon \ell}{2}$ in Claim~\ref{clm:paid-credit-P3a1} and Claim~\ref{clm:diam-cluster-P4a-and-P4b1}. By construction, $\mathcal{D}$ contains at most one edge of $E_i$ which is $e$ (in this case $X$ has at least $c\cdot w(e) + \frac{\epsilon \ell}{2}$ credits). Thus, the remaining credits of $X$ and credits from reserved $\epsilon$-clusters and $\mst$ edges in $\mathcal{D}$ are sufficient for maintaining invariant DC1. Since $\dm(\mathcal{C}''(X)) \geq \ell/2$ by Claim~\ref{clm:diam-cluster-P3a}, $c\cdot \dm(\mathcal{C}''(X)) \geq c\ell/2$.

\vspace{5mm}

\noindent \textbf{Case 4: $\mathcal{C}(X)$ is formed in Case 2 of Phase 3a or in Phase 3b.}    Recall in Claim~\ref{clm:P3a2-P3b-credit-center}, we argue that the center of cluster $X$ collects at least $2c(1 - g\epsilon)\ell$ credits. By construction, $\mathcal{D}$ can contain at most one edge of $E_i$, which connects two cluster paths in  Case 2 of Phase 3a or Phase 3b. We observe that credits in $X$ is taken totally by at most $c(1 - 3g\epsilon - \epsilon/3) \ell$ in Claim~\ref{clm:paid-credit-P3a2-P3b} and Claim~\ref{clm:paid-credit-P4a-P4b1}. Thus, $X$ has at least:
\begin{equation*}
c(2 - 2g\epsilon) \ell  - c(1 - 3g\epsilon - \epsilon/3) \ell > c\ell
\end{equation*}
 remaining credits. That implies the remaining credits of $X$ and credits from reserved $\epsilon$-clusters and $\mst$ edges in $\mathcal{D}$ are sufficient for maintaining invariant DC1. Since $\dm(\mathcal{C}''(X)) \geq \ell/2$ by Claim~\ref{clm:diam-cluster-P3b}, $c\cdot \dm(\mathcal{C}''(X)) \geq c\ell/2$.
\end{proof}

\begin{proof}[Proof of Lemma~\ref{lem:main-lightness}] Recall in the beginning of Phase 4, we assume that $\mathcal{C} \not= \emptyset$ after Phase 3 and in this case, we already paid for every edges of $E_i$  with:
\begin{equation*}
c = \max(\frac{\Theta(g) }{\epsilon^3},\frac{\Theta(\avgH)}{\epsilon} ) = O(\frac{\avgH}{\epsilon^3})
\end{equation*}
and $\epsilon$ sufficiently small.
 
We only need to consider the case when  $\mathcal{C} = \emptyset$ after Phase 3. We have:
\begin{observation}\label{obs:empty-case}
The case when $\mathcal{C} = \emptyset$ after Phase 3 only happens when: (i) there is a single cluster-path $\mathcal{P}$ that contains all $\epsilon$-clusters, (ii) every edge of $E_i$ is incident to an $\epsilon$-cluster in an affix of $\mathcal{P}$ of effective diameter at most $2\ell$ and (iii) $\epsilon$-clusters of $\mathcal{P}$ are low-degree in $\mathcal{K}(\mathcal{C}_\epsilon, E_i)$. 
\end{observation}

We greedily break $\mathcal{P}$ into subpath of $\epsilon$-clusters of effective diameter at least $2\ell$ and at most $4\ell$ as in Phase 4b and form a new cluster from each subpath. Recall a long cluster is formed from a subpath containing at least $\frac{2g}{\epsilon} + 1$ $\epsilon$-clusters. Let $\mathcal{P}'$ be a subpath of $\mathcal{P}$. If $\mathcal{P}'$ is long, we can both buy edges of $E_i$ incident to $\epsilon$-clusters of $\mathcal{P}'$ and maintain two diameter-credit invariants as in Claim~\ref{clm:paid-dm-crd-long-clustesr}. If $\mathcal{P}'$ is short, we use credits of $\epsilon$ and $\mst$ edges of $\mathcal{P}'$ to maintain DC1. Recall  $\mathcal{P}'$ has effective diameter at least $\ell$, thus, has at least $c\ell$ credits by DC1 for level ${i-1}$. That implies $c\cdot \dm(\mathcal{P}') \geq c\max(\frac{\ell}{2}, \dm(\mathcal{P}'))$.

We put remaining unpaid edges of $E_i$ to the  holding bag $B$. Recall unpaid edges of  $E_i$  must be incident to $\epsilon$-clusters of short clusters, which are affixes of $\mathcal{P}$. By Observation~\ref{obs:empty-case}, $B$ holds at most $O(\frac{g}{\epsilon^2}) = O(\frac{1}{\epsilon^2})$ edges of $E_i$. Thus, the total weight of edges of $B$ in all  levels is at most:
\begin{equation}
\begin{split}
O\left(\frac{1}{\epsilon^2}\right)\sum_{i} \ell_i &\leq O\left(\frac{1}{\epsilon^2}\right) \ell_{\max} \sum_{i}\epsilon^i,\ \mbox{where $\ell_{\max} = \max_{e\in S}\{w(e)\}$} \\
    &\leq  O\left(\frac{1}{\epsilon^2}\right)  w(\mst)  \sum_{i} \epsilon^i \\
    & \leq  O\left(\frac{1}{\epsilon^2}\right)  w(\mst) \frac{1}{1-\epsilon} = O\left(\frac{1}{\epsilon^2}w(\mst)\right)
\end{split}
\end{equation}
\end{proof}

\paragraph{Acknowledgments:} 
\bibliographystyle{plain}
\bibliography{spanner}

\appendix

\section{Notation and definitions} \label{sec:prel}

Let $G(V(G),E(G))$ be a connected and undirected graph with a positive edge weight function $w : E(G) \rightarrow \Re^+\setminus \{0\}$. We denote $|V(G)|$ and $|E(G)|$ by $n$ and $m$, respectively. Let $\MST(G)$ be a minimum spanning tree of $G$; when the graph is clear from the context, we simply write $\mst$.  A walk of length $p$ is a sequence of alternating vertices and edges $\{v_0,e_0,v_1,e_1,\ldots, e_{p-1}, v_{p}\}$ such that $e_i = v_iv_{i+1}$ for every $i$ such that $1\leq 0 \leq p-1$. A path is a \emph{simple walk} where every vertex appears exactly once in the walk. 
For two vertices $x,y$ of $G$, we use $d_G(x,y)$ to denote the shortest distance between $x$ and $y$. 

Let $S$ be a subgraph of $G$. We define $w(S) = \sum_{e \in E(S)}w(e)$. 

Let $X \subseteq V(G)$ be a set of vertices. We use $G[X]$ to denote the subgraph of $G$ induced by $X$. Let $Y\subseteq E(G)$ be a subset of edges of $G$. We denote the graph with vertex set $V(G)$ and edge set $Y$ by  $G[Y]$.

We call a graph $K$ \emph{a minor} of $G$ if $K$ can be obtained from $G$ from a sequences of edge contraction, edge deletion and vertex deletion operations. A graph $G$ is $H$-minor-free if it excludes a fixed graph $H$ as a minor.  If $G$ excludes a fixed graph $H$ as a minor, it also excludes the complete $h$-vertex graph $K_h$ as a minor where $h = |V(H)|$.

\begin{observation} \label{obs:minor-sub-div} If a graph $G$ excludes $K_h$ as a minor for $h\geq 3$, then any graph obtained from $G$ by subdividing an edge of $G$ also excludes $K_h$ as a minor.
\end{observation}
\begin{proof} We can assume that $G$ is connected. If $h = 3$, $G$ is acyclic and the observation follows easily.  Let $K$ be the graph obtained from $G$ by subdividing an arbitrary edge, say $e$, of $G$. Let $v$ be the subdividing vertex. Suppose that $K$ contains $K_h$ as a minor. Then  there are $h$ vertex-disjoint trees $\{T_1,T_2,\ldots, T_h\}$ that are subgraphs of $K$ such that each $T_i$ corresponds to a vertex of the minor $K_h$ and there is an edge connecting every two trees. We say $\{T_1,\ldots, T_h\}$ witnesses the minor $K_h$ in $K$. If $v \not\in V(T_1\cup \ldots \cup T_h)$, then $\{T_1,\ldots, T_h\}$ witnesses $K_h$ in $G$, contradicts that $G$ excludes $K_h$ as a minor.  Thus, we can assume,~w.l.o.g, $v \in T_1$. Since $h \geq 4$ and $v$ has degree $2$, $T_1\setminus \{v\} \not= \emptyset$. By contracting $v$ to any of its neighbors in $T_1$, we get a set of $h$ trees witnessing the minor $K_h$ in $G$, contradicting that $G$ is $K_h$ minor-free. 
\end{proof}

\section{Greedy spanners} \label{sec:greed}

A subgraph $S$ of $G$ is a $(1+\epsilon)$-spanner of $G$ if $V(S) = V(G)$ and  $d_S(x,y) \leq (1+\epsilon)d_G(x,y)$ for all $x,y\in V(G)$. The following greedy algorithm by Alth\"ofer et al.~\cite{ADDJS93} finds a $(1+\epsilon)$-spanner of $G$:

\begin{tabbing}
  {\sc GreedySpanner}$(G(V,E), \epsilon)$\\
  \qquad \= $S \leftarrow (V, \emptyset)$.\\
  \> Sort edges of $E$ in non-decreasing order of weights.\\
  \> For each edge $xy \in E$ in sorted order\\
  \> \qquad \=  if $(1+\epsilon)w(xy) < d_S(x,y)$\\
  \>\>\qquad \=  $E(S) \leftarrow E(S) \cup \{e\}$\\
  \qquad \= return $S$
\end{tabbing}

\noindent Observe that as algorithm {\sc GreedySpanner} is a relaxation of  Kruskal's algorithm, $\MST(G) = \MST(S)$. Since we only consider $(1+\epsilon)$-spanners in this work, we simply call an $(1 + \epsilon)$-spanners \emph{a spanner}. We define the lightness of a spanner $S$ to be the ratio $\frac{w(S)}{w(\MST(G))}$. We call $S$ \emph{light} if its lightness is independent of the number of vertices or edges of $G$. 

\section{Reduction to unit-weight $\mst$ edges} \label{app:reduction}

We adapt the reduction technique of Chechik and Wulff-Nilsen~\cite{CW16} to analyze the increase in lightness due to this simplification for $H$-minor-free graphs. Let $G$ be the input graph and let $w : E(G) \rightarrow \Re^+$ be the edge weight function for $G$.  Let $\bar w = \frac{w(\mst)}{n-1}$ be the average weight of the $\mst$ edges. We do the following:
\begin{enumerate}
\item Round up the weight of each edge of $E(G)$ to an integral multiple of $\bar w$.\label{step-round}
\item Subdivide each $\mst$ edge so that each resulting edge has weight exactly $\bar w$.  Let $G'$ be the resulting graph.\label{step-subdivide}
\item Scale down the weight of every edge by $\bar w$.  Let $w'$ be the resulting edge weights of $G'$. $G'$ is minor-free by Observation~\ref{obs:minor-sub-div}. \label{step-scale}
\item Find a $(1+\epsilon)$-spanner $S'$ of $G'$.
\item Let $S$ be a graph on $V(G)$ with edge set equal to the union of $E(S')\cap E(G)$, the edges of $\mst(G)$, and every edge $e$ in $G$ of weight $w(e) \le \frac{\bar w}{\epsilon}$.\label{step-return}
\end{enumerate}

\begin{lemma} \label{lm:unit-ms}
If $S'$ is a $(1+\epsilon)$-spanner of $G'$ with lightness $f(\epsilon)$, then $S$ is a $(1 + O(\epsilon))$-spanner of $G$ with lightness $2f(\epsilon) + O(\avgH/\epsilon)$.
\end{lemma}
\begin{proof}
We adapt the proof of Chechik and Wulff-Nilsen~\cite{CW16}. 

We first bound $w(S)$.  For an edge $e$ in $E(S') \cap E(G)$, $w(e) \le \bar w \cdot w'(e)$ since weights are rounded up before scaling down. Since $G$ is $H$-minor-free, $G$ has $O(\avgH n)$ edges and so has $O(\avgH n)$ edges of weight at most $\frac{\bar w}{\epsilon}$. Thus, the weight of the edges returned in Step~\ref{step-return} is:
\[
w(S)\leq \bar w \cdot w'(S') + w(\mst(G)) + \frac{\bar w}{\epsilon} \cdot O(\avgH n)
\]
Since $S'$ has lightness $f(\epsilon)$ and since $w(\mst(G)) = (n-1)\bar w$, we get
\[
w(S)\le \bar w \cdot f(\epsilon) w'(\mst(G')) + O(\avgH/\epsilon) \cdot w(\mst(G))
\]
The $\mst$ of $G'$ is comprised of the subdivided edges (Step~\ref{step-subdivide}) of the $\mst$ of $G$.  Since the weight of each edge of $\mst(G)$ is rounded up to an integral multiple of $\bar w$, at most $(n-1)\bar w = w(\mst(G))$ is added to the weight of the $\mst$ of $G$.  Therefore $w'(\mst(G')) \le 2 w(\mst(G)) / \bar w$, giving
\[
w(S)\le 2f(\epsilon)\cdot w(\mst(G)) + O(\avgH/\epsilon) \cdot w(\mst(G))
\]
This proves the bound on the lightness of $S$.

We next show that $S$ is a $(1+O(\epsilon))$-spanner of $G$.  It is sufficient to show that for any edge $e \notin E(S)$ there is a path in $S$ of weight at most $(1+O(\epsilon))w(e)$.  Since $S$ contains all edges of weight at most $\frac{\bar w}{\epsilon}$, we may assume that $w(e) > \frac{\bar w}{\epsilon}$.  Let $S'_e$ be a path in $S'$ between $e$'s endpoints of length at most $(1+\epsilon)w'(e)$.  Let $S_e$ be the path in $S$ that naturally corresponds to the path $S'_e$.  As above, we have $w(S_e) \le \bar w \cdot w'(S'_e)$.  Therefore $w(S_e) \le (1+\epsilon) \bar w \cdot w'(e)$.  Since edge weights are rounded up by at most $\bar w$,   $\bar w \cdot w'(e) \le w(e) + \bar w$ which in turn is $\le w(e) + \epsilon w(e)$ since $w(e) > \frac{\bar w}{\epsilon}$.  We get
\[ w(S_e) \le (1+\epsilon)^2 w(e) = (1+O(\epsilon)) w(e). \qedhere\]
\end{proof}

By Lemma~\ref{lm:unit-ms}, we may assume that all edges of $\mst(G)$ have weight $1$. We find the $(1+\epsilon)$-spanner $S$ of $G$ by using the greedy algorithm. Thus, the stretch condition of $S$ is satisfied.

\end{document}